\documentclass[a4paper,reqno,11pt]{amsart}

\usepackage[hmargin=3cm,vmargin=3cm]{geometry}

\usepackage{amssymb,amsthm,mathtools,cancel,mathrsfs,mathdots,color,graphicx,framed,enumerate,tikz}
\usepackage[smalltableaux]{ytableau}

\usepackage[colorlinks=true, pdfstartview=FitV, urlcolor=blue, citecolor=red, linkcolor=blue]{hyperref}
\usepackage{stmaryrd}
\SetSymbolFont{stmry}{bold}{U}{stmry}{m}{n} 

\usepackage[T1]{fontenc}
\usepackage[landscape=true]{pdflscape}

\def\Z{\mathbb{Z}}

\def\Q{\mathbb{Q}}
\def\F{\mathbb{F}}
\def\wt{\widetilde}
\def\wh{\widehat}

\def\d{\mathrm d}
\def\e{\mathrm{e}}
\def\i{\mathrm{i}}
\def\pa{\partial}
\def\p{\bs p}

\newcommand{\SSF}{\Lambda^*}

\newcommand{\NO}[1]{\vcentcolon\mathrel{#1}\vcentcolon}
\providecommand{\vcentcolon}{\mathrel{\mathop{:}}}

\def\be{\begin{equation}}
\def\ee{\end{equation}}

\newtheorem{theorem}{Theorem}[section]
\newtheorem{conjecture}[theorem]{Conjecture}
\newtheorem{example}[theorem]{Example}
\newtheorem{lemma}[theorem]{Lemma}
\newtheorem{proposition}[theorem]{Proposition} 
\newtheorem{corollary}[theorem]{Corollary}

\theoremstyle{remark}
\newtheorem{remark}[theorem]{Remark}

\newcommand{\partitions}{\mathscr{P}}
\newcommand{\pdv}[2]{\frac{\partial #1}{\partial #2}}

\def\fq#1{\left\langle{#1}\right\rangle_{\! q}}

\def\bs#1{\boldsymbol{#1}}

\renewcommand{\=}{\: =\: }
\newcommand{\+}{\,+\,}
\newcommand{\meno}{\,-\,}

\allowdisplaybreaks

\begin{document}

\numberwithin{equation}{section}

\title[Quantum KdV hierarchy and shifted symmetric functions]{Quantum KdV hierarchy \\ and shifted symmetric functions}

\author{Jan-Willem M. van Ittersum}
\address{Max-Planck-Institut f\"ur Mathematik, Vivatsgasse 7, 53111 Bonn, Germany}
\curraddr{Department of Mathematics and Computer Science, University of Cologne,
Weyertal 86-90, 50931 Cologne, Germany}

\email{j.w.ittersum@uni-koeln.de}

\author{Giulio Ruzza}
\address{Departamento de Matem\'atica, Faculdade de Ci\^encias da Universidade de Lisboa, Campo Grande Edif\'{i}cio C6, 1749-016, Lisboa, Portugal}
\email{gruzza@fc.ul.pt}

\date{}

\begin{abstract}
We study spectral properties of the quantum Korteweg--de Vries hierarchy defined by Buryak and Rossi.
We prove that eigenvalues to first order in the dispersion parameter are given by shifted symmetric functions.
The proof is based on the boson-fermion correspondence and an analysis of quartic expressions in fermions.
As an application, we obtain a closed evaluation of certain double Hodge integrals on the moduli spaces of curves.
Finally, we provide an explicit formula for the eigenvectors to first order in the dispersion parameter. In particular, we show that its Schur expansion is supported on partitions for which the Hamming distance is minimal.
\end{abstract}


\maketitle

\section{Introduction and results}

\subsection{Generating functions of double ramification intersection numbers}

The quantum double ramification hierarchy is a remarkable construction by Buryak and Rossi~\cite{BR}, also inspired by the symplectic field theory program of Eliashberg, Givental and Hofer~\cite{EGH}, which associates a quantum integrable hierarchy to an arbitrary cohomological field theory~\cite{KM}.
This construction was an important step forward in studying relations between the geometry of moduli spaces of curves and the theory of integrable systems, initiated with Witten's seminal conjecture~\cite{Witten}.
Double ramification hierarchies have been extensively studied in recent years, cf.~\cite{Blot,BB,BDGR,BGR,BLS}.
In this paper, we are interested in the spectral problem for 
this hierarchy.

Let~$\wt{\mathcal B}$ be the $\Q$-algebra with generators~$\omega_a$ ($a\in\Z$) and~$\hbar$ and commutation relations
\be
\label{eq:ccr}
[\hbar,\omega_a]=0,\quad[\omega_a,\omega_b]\=-a\,\delta_{a,-b}\,\hbar\qquad(a,b\in\Z),
\ee
where we denote by $[A,B]=AB-BA$ the commutator. Let us introduce the normal ordering
\be
\label{eq:no}
\NO{\omega_{a_1}\dots \omega_{a_n}}\=\prod_{a_i\geq 0}\omega_{a_i}\prod_{a_i<0}\omega_{a_i}.
\ee
In~\cite{BR}, for $k\geq 0$, Buryak and Rossi consider generating functions 
\be
\label{eq:gk}
g_k(\epsilon;z) \= \sum_{\substack{g,n\geq 0 \\ 2g-1+n>0}}\frac{\hbar^g}{n!} \sum_{a_1,\dots,a_n\in\mathbb Z} I_{g,k;a_1,\dots,a_n}\biggl(\frac \epsilon\hbar\biggr)\, z^{\sum_\ell a_\ell}\,\NO{\omega_{a_1}\dots\omega_{a_n}}
\ee
of intersection numbers
\be
\label{eq:I}
I_{g,k;a_1,\dots,a_n}(y)\=\int_{\mathrm{DR}_g(-\sum_\ell a_\ell,a_1,\dots,a_n)}\psi_1^k \,\bigl(1+y\lambda_1+y^2\lambda_2+\dots+y^g\lambda_g\bigr)
\ee
of~$\psi$ and~$\lambda$ classes over the double ramification cycle $\mathrm{DR}_g(-\sum_l a_l,a_1,\dots,a_n)$ in the moduli spaces of curves~$\overline{\mathcal M}_{g,n+1}\mspace{1mu}$. (Here,~$y$ is a formal variable.)
Note that $g+n\leq k+2$ in the right-hand side of Eq.~\eqref{eq:gk} for dimensional reasons.
For more details, we refer the reader to~\cite{B,BR,BSSZ,JPPZ}.

To deal with the infinite sums in the definition of these generating functions in Eq.~\eqref{eq:gk}, we need to work with a completion~$\mathcal B$ of~$\wt{\mathcal B}$ (see~\cite{BR,EGH}), whose elements are (formal) power series in $\omega_{-1}, \omega_{-2},\ldots$ with coefficients which are polynomials in $\hbar,\omega_0,\omega_1,\omega_2,\dots$; the product of two such power series can always be brought into the form of another such power series by use of the commutation relations~\eqref{eq:ccr}.
Then, $g_k(\epsilon;z)\in\mathcal B[\epsilon]\llbracket z^{\pm 1}\rrbracket$.

\begin{remark}\label{rem:normalizationdiff}
The variables $\epsilon,\hbar,z,\omega_a$ of the present paper coincide with $-\varepsilon^2,\i \hbar,\e^{-\i x},p_{-a}$ in~\cite{BR}.
As we will recall below, $I_{g,k;a_1,\dots,a_n}(y)$ is an even polynomial in the variables $a_i$.
\end{remark}

The results of~\cite{BR} give interesting properties of these generating functions (generalizing to arbitrary cohomological field theories).
In particular, in op.~cit.\ it is shown that $g_k(\epsilon;z)$ can be computed by an explicit recursion in~$k$ (which we recall in Section~\ref{sec:rec} below). Moreover,  
\be
G_k(\epsilon)\=\mathrm{Res}_z\, g_k(\epsilon;z)\frac{\d z}z \in \mathcal B[\epsilon]\qquad (k\geq 0)
\ee
where $\mathrm{Res}_z\,\sum_{n\in\Z}A_n\,z^n\,\d z=A_{-1}$ is the formal residue, commute pairwise, i.e.,
\be
\label{eq:integrability}
[G_k(\epsilon), G_\ell(\epsilon)]=0\qquad (k,\ell\geq 0).
\ee
The pairwise commuting elements~$G_k(\epsilon)$ form the quantum double ramification hierarchy for the trivial cohomological theory.
\begin{example}
\label{ex1}
When $k=0$, only terms with $(g,n)=(0,2)$, $(1,1)$, or $(1,0)$ contribute to the sum in Eq.~\eqref{eq:gk}.
These terms are readily computed {\upshape (}cf.{\upshape~\cite{B,BR,JPPZ})} and we have
\be
g_0(\epsilon;z)\=\frac 12\sum_{a_1,a_2\in\Z}\NO{\omega_{a_1}\omega_{a_2}}z^{a_1+a_2}\+\frac {\epsilon}{24}\sum_{a\in\Z}a^2\omega_a z^a \meno \frac \hbar {24}.
\ee
Similarly,
\begin{multline}
g_1(\epsilon;z)\=\frac 16\sum_{a_1,a_2,a_3\in\Z}\!\!\NO{\omega_{a_1}\omega_{a_2}\omega_{a_3}}z^{a_1+a_2+a_3}\+\frac\epsilon{24}\sum_{a_1,a_2\in\Z}\!\!\frac{a_1^2+a_2^2}2\NO{\omega_{a_1}\omega_{a_2}}z^{a_1+a_2}
\\
\label{eq:g1example}
\+\frac{\epsilon^2}{1152}\sum_{a\in\Z}a^4\omega_a\, z^a\+\frac{\hbar}{24}\sum_{a\in\Z}(a^2-1)\omega_a \,z^a\+\frac{\epsilon\hbar}{2880}.
\end{multline}
\end{example}

\begin{remark}[Quantum KdV]
In the classical limit, $\hbar=0$ and the variables~$\omega_a$ ($a\in\Z$) form a commutative algebra.
Then, $G^{\mathrm{cl}}_k(\epsilon) = \lim_{\hbar \to 0} G_k(\epsilon)$ are the conserved quantities of the classical Korteweg--de\,Vries (KdV) hierarchy, cf.~\cite{B} ($G_1^{\mathrm{cl}}$ being the KdV Hamiltonian). 
We have $\bigl\lbrace G_k^{\mathrm{cl}}(\epsilon),G_\ell^{\mathrm{cl}}(\epsilon)\bigr\rbrace=0$, where the Poisson bracket is defined by $\bigl\lbrace\omega_a,\omega_b\bigr\rbrace\=\i \,a\,\delta_{a,-b}\mspace{1mu}$, which is the {\it first} Poisson structure of the KdV hierarchy (in Fourier coordinates).
Hence, the quantum double ramification hierarchy described above solves the quantization problem for the KdV hierarchy by explicitly providing an $\hbar$-deformation of the classical KdV conserved quantities which form a commuting family of quantum conserved charges, which we call the {\it quantum KdV hierarchy}.
We direct the reader to~\cite[Section~2]{BR} and~\cite[Section~1]{Dubrovin} and references therein for more details.

It is worth noting that the classical KdV hierarchy also admits a {\it second} Poisson structure, whose quantization involves the Virasoro algebra and has been extensively studied in Conformal Field Theory after the impetus of the seminal work by Bazhanov, Lukyanov, and Zamolodchikov~\cite{BLZ}.
\end{remark}

\subsection{Spectral problem}

Consider the representation $\rho_{c}:\mathcal B\to\mathrm{End}\,\mathsf B$, where $\mathsf B=\Q[\bs p]$, $\bs p=(p_1,p_2,\dots)$, defined by 
\be
\label{eq:repB}
\rho_{c}(\hbar) \, \phi\= \phi,\quad
\rho_{c}(\omega_a) \, \phi\= \begin{cases}
p_a\, \phi, & \mbox{if }a>0, \\
c \,\phi, & \mbox{if } a=0,\\
-a\frac{\partial \phi}{\partial p_{-a}}, & \mbox{if }a<0,
\end{cases}\qquad (\phi\in\mathcal{B},\, a\in\mathbb Z)\,.
\ee
depending on an arbitrary rational number~$c$.

\begin{remark}
We could let~$\hbar$ act by multiplication by an arbitrary nonzero constant.
Without loss of generality, we set this constant to~$1$, as we can always reduce to this case by rescaling the variables~$p_k\mspace{1mu}$.
\end{remark}

We obtain the family of commuting \emph{quantum KdV Hamiltonian operators} on~$\mathsf B$
\be
\wh G_k(\epsilon,c) = \rho_{c}\bigl(G_k(\epsilon)\bigr)\qquad (k\geq 0),
\ee
polynomially depending on~$\epsilon$ and~$c$.
It is readily checked that these operators preserve the grading $\mathsf B=\bigoplus_{n\geq 0}\mathsf B_n$ induced by assigning weight~$k$ to the variable~$p_k\mspace{1mu}$.

In this paper, we continue the study of the spectral problem for such a family of commuting operators on~$\mathsf B$.
Rossi~\cite{Rossi} (see also the work of Pogrebkov~\cite{Pogrebkov}) showed that $\wh G_k(0;c)$ is quadratic in fermions under the boson-fermion correspondence, which we will recall in Section~\ref{sec:BF}.
Dubrovin then showed that, when $\epsilon=0$, a basis of common eigenvectors is given by Schur functions.
Namely,~\cite[Corollary~2.4]{Dubrovin}
\be
\label{eq:Dubrovin}
\wh G_k(0;c)\, s_\lambda \= E_k^{[0]}(\lambda;c)\, s_\lambda\qquad (k\geq 0),
\ee
where the Schur functions $s_\lambda=s_\lambda(\p)$ are labeled by partitions $\lambda=(\lambda_1,\ldots,\lambda_r)\in \partitions$ and are defined by (cf.~\cite{Mac})
\be
\label{eq:schur}
s_\lambda(\bs p) \= \det\bigl(h_{\lambda_i-i+j}(\bs p)\bigr)_{i,j=1}^{r}\,,\quad \sum_{\ell\in\Z} h_\ell(\bs p)\,y^\ell \= \exp\Bigl(\sum\nolimits_{\ell\geq 1}\frac{p_\ell}\ell y^\ell\Bigr).
\ee
Moreover, Dubrovin showed that the eigenvalue in Eq.~\eqref{eq:Dubrovin} is a {\it shifted symmetric} function on partitions, that is
\be
\label{eq:E0}
E_k^{[0]}(\lambda;c) \= \sum_{j=0}^{k+2}\frac{c^{k+2-j}}{(k+2-j)!}\,Q_j(\lambda)\hspace{1pt},
\ee
where the functions $Q_k:\partitions\to \Q$ are given by $Q_0(\lambda)=1$ and 
\be
\label{eq:Qk}
Q_k(\lambda) \= \frac 1{(k-1)!}\sum_{i\geq 1}\left[\left(\lambda_i-i+\tfrac{1}{2}\right)^{k-1}-\left(-i+\tfrac{1}{2}\right)^{k-1}\right]\+\beta_k\qquad (k\geq 1)
\ee
with
\be
\label{eq:betagen}
\sum_{k\geq 0}\beta_k\,y^k=\frac{y/2}{\sinh(y/2)}.
\ee
Equivalently, $\beta_k \= \frac{1}{k!}\bigl(\frac 1{2^{k-1}}-1\bigr)B_k\mspace{1mu}$,  where~$B_k$ denotes the~$k$th Bernoulli number.

\begin{remark}\label{rem:ss}
The algebra of shifted symmetric functions is a deformation of the algebra of symmetric functions, for which the natural generators are the~$Q_k\mspace{1mu}$. A vector space basis is given by the central characters of the symmetric group. Shifted symmetric functions appear in asymptotic representation theory \cite{IvanovOlshanski,KerovOlshanski,RiosZertuche} and in enumerative geometry, e.g., in the Hurwitz/Gromov--Witten theory of curves \cite{Dijkgraaf, Engel, EskinOkounkov, HahnIttersumLeid, Ochiai,OkounkovPandharipande}, and in computations of volumes and Siegel--Veech constants of the moduli space of flat surfaces \cite{ChenMollerZagier,ChenMollerSauvagetZagier,EskinOkounkov}. For the latter, a remarkable relation to quasimodular forms~\cite{BO,Zagier} is crucial, which will be recalled below (cf.~Theorem~\ref{thm:BO}).
\end{remark}

Furthermore, it was shown in~\cite{RY} that the commuting Hamiltonians~$\wh{G}_k$ admit a common basis of eigenvectors deforming the Schur polynomials. That is, for each partition~$\lambda$ of size $|\lambda|=\sum_{i}\lambda_i$ there exists a unique power series $r_\lambda(\bs p;\epsilon)\in \mathsf B_{|\lambda|}\llbracket\epsilon\rrbracket$  satisfying
\be
\label{eq:spectrum}
\wh G_k(\epsilon,c)\,r_\lambda(\bs p;\epsilon)\= E_k(\lambda;\epsilon,c)\, r_\lambda(\bs p;\epsilon)
\ee
such that
\begin{align}
\label{eq:spectrum1}
r_\lambda(\bs p;\epsilon)&=s_\lambda+\sum_{m\geq 1}\epsilon^mr_\lambda^{[m]}(\bs p),\\
E_k(\lambda;\epsilon,c)&=E^{[0]}_k(\lambda;c)+\sum_{m\geq 1}\epsilon^mE_k^{[m]}(\lambda;c)
\end{align}
and
\be
\label{eq:gauge}
\langle s_\lambda,r_\lambda^{[m]}\rangle=0 \qquad (m\geq 1),
\ee
where the standard scalar product on~$\mathsf B$ is given by $\langle s_\lambda(\bs p),s_\mu(\bs p)\rangle=\delta_{\lambda,\mu}\mspace{1mu}$.

\subsection{Results}

Our first result expresses the eigenvalue~$E_k^{[1]}(\lambda;c)$ as a shifted symmetric function in the partition~$\lambda$.

\begin{framed}
\begin{theorem}\label{thm:1}
For all $k\geq 0$ we have
\be
\label{eq:E1}
E_k^{[1]}(\,\cdot\,;c) \:=\: \frac 1{24}\sum_{\ell=0}^k\frac{c^\ell}{\ell!}\biggl(2Q_2Q_{k+1-\ell}\+(k-\ell)(k-\ell+3)Q_{k+3-\ell}\biggr)
\ee
where the shifted symmetric function~$Q_k$ is defined by Eq.~\eqref{eq:Qk}.
\end{theorem}
\end{framed}
The proof is given in Section~\ref{sec:proof1} and employs the boson-fermion correspondence.
A challenging aspect of this proof is the fact that the linear part in~$\epsilon$ of the quantum KdV Hamiltonian operators are quartic expressions in fermions, rather than simpler quadratic expression, cf.~\cite{Dubrovin,Rossi}.

We obtain the following application.
\begin{framed}
\begin{corollary}
\label{corollary:doubleHodge}
We have the following explicit evaluation of double Hodge intersection numbers on the moduli spaces of curves:
\be
\label{eq:in}
\int_{\overline{\mathcal M}_{g,1}}\lambda_1\lambda_g\psi_1^{2g-3} \:=\: \frac{(-1)^g}{24}\biggl(2g(2g-3)\beta_{2g}-\frac 1{12}\beta_{2g-2}\biggr)\qquad (g>1),
\ee
where~$\beta_k$ is defined by Eq.~\eqref{eq:betagen}.
\end{corollary}
\end{framed}

For our second result, we focus on the eigenvectors in Eq.~\eqref{eq:spectrum}, also to first order in~$\epsilon$.
We show that $\langle r_\lambda^{[1]},s_\mu\rangle \neq 0$ if and only if the Hamming distance~$d$, given by Eq.~\eqref{eq:Hamming}, of the Frobenius coordinates of two partitions~$\lambda$ and~$\mu$ is exactly~$2$ (which is the minimal distance for two partitions of the same integers).
If this is the case,~$\lambda$ can be obtained from~$\mu$, in exactly two ways, by removing a border strip~$\gamma_i$ and adding a border strip~$\gamma_i'$ ($i=1,2$, with $|\gamma_i|=|\gamma_i'|$, $\gamma_1\subset\gamma_2$, $\gamma_1'\subset\gamma_2'$).
(More detailed explanations and proofs are given in Sections~\ref{sec:backgroundBF} and~\ref{introrefsecthm2}.)

\begin{framed}
\begin{theorem}
\label{thm:2}
For all partitions~$\lambda$
\be
\label{eq:V1}
r^{[1]}_\lambda(\bs p)\,= \frac 1{12} \sum_{d(\lambda,\mu)=2}w_{\lambda,\mu}\, (-1)^{\mathrm{ht}(\gamma_1)+\mathrm{ht}(\gamma_1')} \,\Biggl(\frac{|\gamma_1|}{|\gamma_2|}-\frac{|\gamma_2|}{|\gamma_1|}\Biggr)\, s_{\mu}(\bs p)\,.
\ee
Here $w_{\lambda,\mu}=1$ if $\lambda_i>\mu_i$ for the smallest $i\geq 1$ such that $\lambda_i\not=\mu_i$, and $w_{\lambda,\mu}=-1$ otherwise. 
Moreover,~$\mathrm{ht}(\gamma)$ and~$|\gamma|$ are, respectively, the height and size of a border strip~$\gamma$.
\end{theorem}
\end{framed}

\subsection{Further comments and conjectures}

We denote by $\SSF=\Q[Q_2, Q_3,\dots]$ the algebra of shifted symmetric functions on partitions. 
The expressions for $E^{[0]}_k(\lambda;c)$ in Eq.~\eqref{eq:E0} and for~$E_k^{[1]}(\lambda;c)$ in Eq.~\eqref{eq:E1} prompt a conjecture, suggested to us by Don Zagier and already stated in~\cite{IR}.
Namely, we conjecture that~$E_k^{[m]}(\lambda;c)$ belongs to~$\SSF[c]$. More precisely, by assigning weight~$k$ to~$Q_k$ and weight~$1$ to~$c$ the algebra $\SSF[c]=\bigoplus_{n\geq 0}\SSF[c]_{n}$ becomes a graded algebra, and we conjecture that $E_k^{[m]}(\lambda;c)$ is of homogeneous weight $k+2+m$ in~$\SSF[c]$ for all $k,m\geq 0$. Based on extensive numerical verification, which we report in Appendix~\ref{app:table}, we formulate the following refined version of that conjecture.
\begin{conjecture}\label{conj}
For each integer $D\geq0$ and each partition $\nu$, there exists a polynomial~$f_{D,\nu}$ with rational coefficients of degree $2D$ such that the following identity holds for all $k\geq 0$:
\be 
E_k(\,\cdot\,;\epsilon,0) \= \sum_{\substack{\nu \in \partitions \\ D\geq 0}}   \epsilon^{D+|\nu|} f_{D,\nu}(k)\, Q_{k+D+2-\ell(\nu)} \prod_{i}  Q_{\nu_i+1}\mspace{1mu}.
\ee
\end{conjecture}

\begin{remark}
Note that by Lemma~\ref{lemma:c0}, we have $E_k(\lambda;\epsilon,c)=\sum_{\ell=0}^k\frac{c^\ell}{\ell!}E_{k-\ell}(\lambda;\epsilon,0)$, such that the restriction to $c=0$ in this conjecture does not imply any loss of generality.
\end{remark}

\begin{example}
We have
\be
f_{0,\emptyset}(k) = 1, \quad f_{1,\emptyset}(k)=\frac{1}{24}k(k+3), \quad \text{and} \quad f_{0,(1)}(k) = \frac{1}{12},
\ee
and with these polynomials we recover Dubrovin's result in Eq.~\eqref{eq:E0}, and Theorem~\ref{thm:1}. 
\end{example}
See Appendix~\ref{app:table} for a table of such polynomials, which give a conjectural expression for~$E_k^{[m]}$ up to $m\leq 5$.
For example, we conjecture 
\begin{multline}\label{eq:conjE2}
E_k^{[2]}(\,\cdot\,;0)\= \frac 1{24^2}\biggl(2 Q_3Q_{k+1}\+2 Q_2^2Q_{k}\+2(k^2+2k-1)Q_2Q_{k+2} \\
\+\frac 1{6}k(k+4)(3k^2+16k+17)Q_{k+4}\biggr)\,.
\end{multline}
We plan to provide an explicit construction for the polynomials~$f_{D,\nu}$ in a future publication in collaboration with Don Zagier.

\begin{remark}
The arguments in the proof of Corollary~\ref{corollary:doubleHodge} can be easily adapted to show that Conjecture~\ref{conj} implies the existence of polynomials~$p_{r,s}$ of degree $2(s-r)$ for all integers $r,s$ with $0\leq r\leq s$ such that
\be
\label{eq:conjHodge}
\int_{\overline{\mathcal M}_{g,1}}\lambda_s\lambda_g\psi_1^{2g-2-s} \:=\:(-1)^g\sum_{r=0}^{s}p_{r,s}(g)\,\beta_{2(g-r)}\qquad (s,g\in \mathbb{N} \text{ with } s<g).
\ee
For example, the conjectural Eq.~\eqref{eq:conjE2} yields
\be
\int_{\overline{\mathcal M}_{g,1}}\!\!\lambda_2\lambda_g\psi_1^{2g-4} \:=\: \frac {(-1)^g}{24^2}\biggl(\frac{\beta_{2g-4}}{288}-\frac{(4g^2-12g+7)\beta_{2g-2}}{12}+\frac {2 g(g-2)(12g^2-16g+1)\beta_{2g}}3\biggr)
\ee
for $g>2$.
It is interesting to compare the expression~\eqref{eq:conjHodge} with known formulas for Hodge integrals, for example 
\cite{Fab,GP}.
\end{remark}

Supporting evidence for the above conjecture, in addition to Theorem~\ref{thm:1} and~Appendix~\ref{app:table}, is provided in our previous work~\cite{IR} through a connection between KdV Hamiltonian operators and quasimodular forms which we briefly review now.

Let us first recall the following theorem of Bloch and Okounkov~\cite{BO} (see also~\cite{Zagier}) concerning the $q$-bracket of functions on partitions, given by
\be
\fq f \= \frac{\sum_{\lambda\in\partitions}f(\lambda)\,q^{|\lambda|}}{\sum_{\lambda\in\partitions}q^{|\lambda|}} \in \mathbb{Q}\llbracket q \rrbracket
\ee
for any function $f:\partitions\to\Q$ on the set~$\partitions$ of all partitions. Write~$\widetilde{M}_k$ for the space of all weight~$k$ quasimodular forms for the full modular group~$\mathrm{SL}_2(\Z)$.

\begin{theorem}[\cite{BO}]\label{thm:BO}
For all $k\geq 0$ and all $f\in\SSF_k$, we have $\langle f \rangle_q \in \widetilde{M}_k\mspace{1mu}.$
\end{theorem}

Because of this theorem, the validity of Conjecture~\ref{conj} would imply the quasimodularity of the $q$-brackets of the KdV eigenvalues. 
This has been independently verified in our previous work.
We recall the relevant result as it will be useful later on.
To this end, we extend the grading on the algebra of quasimodular forms to $\widetilde{M}[c, \epsilon]$ by assigning weight~$1$ to~$c$ and weight~$-1$ to~$\epsilon$.

\begin{theorem}[\cite{IR}]\label{thm:quasimodular}\label{thm:IR}
For all $k\geq 0$, we have $\big\langle {E_k}(\,\cdot\,;\epsilon,c) \big\rangle_q \in \widetilde{M}[c, \epsilon]_{k+2}\mspace{1mu}.$
\end{theorem}

\begin{remark}
An intriguing aspect of this result is that the eigenvalues $E_k(\cdot;\epsilon,0)$ (taking $c=0$ for simplicity) interpolate between the generators $Q_{k}$ of the algebra of shifted symmetric functions at $\epsilon=0$ and the generators $S_k$ of the algebra of symmetric functions (namely, $S_k(\lambda)=-\frac{B_k}{2k}+\sum_i\lambda_i^{k-1}$) at $\epsilon=\infty$ (cf.~\cite[Appendix~A]{IR}).  
Both these algebras of functions on partitions are \emph{quasimodular algebras}, in the sense of~\cite{vI}, namely, $q$-brackets of homogeneous elements are quasimodular forms of pure weight. No other quasimodular algebras are known for the full modular group~$\mathrm{SL}_2(\Z)$.
Then, Theorem~\ref{thm:quasimodular} asserts that this $\epsilon$-deformation respects quasimodularity.
\end{remark}

Finally, given the commutativity~\eqref{eq:integrability}, Conjecture~\ref{conj} would imply the stronger result that for all $k_1,\ldots,k_r\geq 0$, we have
\be
\label{eq:conjeq}
\Big\langle \prod_{i}{E_{k_i}} (\,\cdot\,;\epsilon,c) \Big\rangle_q \in \widetilde{M}[c, \epsilon]_{\sum_i (k_i+2)}\ \raisebox{-5pt}{.}
\ee

Note that when $\epsilon=0$, Eq.~\eqref{eq:conjeq} coincides with the quasimodularity property first established by Okounkov and Pandharipande~\cite[Section~5]{OkounkovPandharipande} for the Gromov--Witten theory of an elliptic curve (see also~\cite{Rossi}).
Then, the conjectural relation~\eqref{eq:conjeq} would imply that this quasimodularity property persists under suitable insertion of Hodge classes.

\subsection*{Outline of the rest of the paper}

In Section~\ref{sec:backgroundqDR} and~\ref{sec:backgroundBF} we recall some preliminary material on quantum double ramification hierarchies, partitions, and the boson-fermion correspondence, needed for proving our main theorems.
Theorem~\ref{thm:1} and Corollary~\ref{corollary:doubleHodge} are proven in Section~\ref{sec:proof1}.
We prove Theorem~\ref{thm:2} in Section~\ref{sec:proof2}.
Finally, Appendix~\ref{app:table} contains supplemental data to Conjecture~\ref{conj}.

\section*{Acknowledgments}
We thank Konstantin Aleshkin, Alexandr Buryak, Giordano Cotti, Davide Masoero, Paolo Rossi, Di Yang, and Don Zagier for valuable discussions.
A large part of this work was completed during our stay at the Max Planck Institute for Mathematics in Bonn, which we wish to thank for the excellent working conditions.

The first author was supported by the SFB/TRR 191 ``Symplectic Structure in Geometry, Algebra and Dynamics'', funded by the DFG (Projektnummer 281071066 TRR 191). He also thanks the Grupo de F\'isica Matem\'atica of Lisbon for their support during a one-week visit to Instituto Superior T\'ecnico, to which he extends his gratitude for their hospitality.

The second author is supported by the FCT grant 2022.07810.CEECIND, and is grateful to the Grupo de F\'isica Matem\'atica (UIDB/00208/2020, UIDP/00208/2020, DOI: 10.54499/UIDB/00208/2020, 10.54499/UIDP/00208/2020) for support.

\section{Quantum double ramification hierarchies}\label{sec:backgroundqDR}

In this section, we recall several aspects of quantum double ramification hierarchies.

\subsection{Recursion}\label{sec:rec}
We recall the following result from~\cite[Theorem~3.5]{BR} (with small changes due to slightly different normalizations, for which we refer to Remark~\ref{rem:normalizationdiff}).
To this end, it is convenient to define
\be
\label{eq:initialdatumBRrec}
g_{-1}(\epsilon;z)\=\sum_{a\in\Z}\omega_a z^a.
\ee

\begin{theorem}[\cite{BR}]
We have 
\be
\label{eq:BRrec}
-\biggl(k+2+\epsilon\frac{\partial}{\partial\epsilon}\biggr)\, z\pa_z\, g_{k+1}(\epsilon;z) \= \frac 1{\hbar}\bigl [g_{k}(\epsilon;z),G_1(\epsilon)\bigr] 
\qquad (k\geq -1)
\ee
and
\be
\label{eq:G1rec}
G_1(\epsilon) \= \,\frac 16\sum_{a_1,a_2\in\Z}\NO{\omega_{a_1}\omega_{a_2}\omega_{-a_1-a_2}}+\frac\epsilon{24}\sum_{a\in\Z}a^2\NO{\omega_{a}\omega_{-a}}-\frac{\hbar}{24}\omega_0+\frac{\epsilon\hbar}{2880}.
\ee
\end{theorem}

\begin{remark}
The formulation in~\cite[Theorem~3.5]{BR} is given in terms of differential polynomials (cf.~Section~\ref{sec:diffpoly}) and it involves the operator $D=2\epsilon\partial_\epsilon+2\hbar\partial_{\hbar}+\sum_{k\geq 0}u_k\partial_{u_k}$, acting on polynomials in $\bs u=(u_0,u_1,\dots)$ also depending on parameters $\epsilon,\hbar$.
We have made use of the dimensional constraint in Eq.~\eqref{eq:I} to rewrite the recursion as in Eq.~\eqref{eq:BRrec}, cf.~\cite[Eq.~(4.7)]{RY}.
\end{remark}
Formula~\eqref{eq:BRrec}, together with Eq.~\eqref{eq:initialdatumBRrec}, allows us to recursively determine $g_k(\epsilon;z)$ up to a constant depending on $\epsilon,\hbar$, but not on the $\omega_a\mspace{1mu}$, for all $k\geq 0$, as explained in~\cite[Section~3.5]{BR}.
Here, to compute the constant term in $z$ of $g_k(\epsilon,z)$ (which in principle is killed by the operator $z\partial_z$ on the left-hand side of the recursion) one uses the fact that the coefficients of Eq.~\eqref{eq:gk} given in Eq.~\eqref{eq:I} are polynomials in $a_1,\dots,a_n\mspace{1mu}$ (see Section~\ref{sec:diffpoly}).
Finally, note that the recursion is insensitive to the actual value of the constant depending on $\epsilon,\hbar$ only, and that the value of such constant can finally be recovered by either the string equation~\cite[Lemma~3.7]{BR}
\be
\label{eq:string}
\frac{\partial g_{k+1}(\epsilon;z)}{\partial \omega_0} \= g_k(\epsilon;z), \qquad (k\geq -1)
\ee
or by invoking the homogeneity expressed by Theorem~\ref{thm:quasimodular}.

It will be convenient to introduce $g_k^{[j]}$ and $G_k^{[j]}$ for $j,k\geq 0$ by
\be
\label{eq:decomp}
g_k(\epsilon;z)\=\sum_{j=0}^{k+1}g_k^{[j]}(z)\biggl(\frac{\epsilon}{24}\biggr)^j,\quad
G_k(\epsilon)\=\sum_{j=0}^{k}G_k^{[j]}\biggl(\frac{\epsilon}{24}\biggr)^j,\qquad 
\ee
such that we can express Eq.~\eqref{eq:BRrec} as
\be
\label{eq:BRrechomogeneous}
{-}\hbar\,(k+2+j)\,z\pa_z\, g_{k+1}^{[j]}(z)\=\bigl[g_k^{[j]}(z),G_1^{[0]}\bigr]\+\bigl[g_k^{[j-1]}(z),G_1^{[1]}\bigr]\qquad (k\geq 0,\ j\geq 0).
\ee
(We agree that $g_k^{[-1]}(z)=0$.)
The upper bounds of summation of the sums in Eq.~\eqref{eq:decomp} stem from the dimensional constraint in Eq.~\eqref{eq:I}, as well as from the fact that ${\rm DR}_g(0)$ is Poincar\'e dual to $(-1)^g\lambda_g\mspace{1mu}$, cf.~\cite[Section~0.5.3]{JPPZ}, combined with the well-known identity $\lambda_g^2=0$ for $g\geq 0$.

\subsection{Reduction to \texorpdfstring{$c=0$}{c=0}}\label{sec:reduction}

In what follows we set $c=0$ in the representation $\rho_c$ defined by Eq.~\eqref{eq:repB}.
This is without loss of generality, as we now show.

\begin{lemma}\label{lemma:c0}
We have
\be
\wh G_k(\epsilon,c) \=\sum_{\ell=0}^{k}\frac{c^\ell}{\ell!}\,\wh G_{k-\ell}(\epsilon,0).
\ee
\end{lemma}
\begin{proof}
By Eq.~\eqref{eq:repB}, for all $X\in\mathcal B$ we have $\rho_c(X) \= \rho_0\bigl(\e^{c\frac{\partial}{\partial\omega_0}}X\bigr)$.
The proof follows by taking $X=G_k(\epsilon)$ and recalling Eq.~\eqref{eq:string}.
\end{proof}

In view of this lemma, Theorem~\ref{thm:1} is equivalent to the identity
\be
\label{eq:E1simple}
E_k^{[1]}(\lambda;0) \= \frac 1{24}\left( 2\,Q_2(\lambda)\,Q_{k+1}(\lambda)\+k(k+3)\,Q_{k+3}(\lambda)\right)\,.
\ee

\subsection{Differential polynomials formalism}\label{sec:diffpoly}
Finally, we recall that there is an equivalent description of the generating functions $g_k(\epsilon;z)$ in terms of differential polynomials.
To explain it, we first recall that the integral of any tautological class against the double ramification cycle $\mathrm{DR}_g(-\sum_l a_l,a_1,\dots,a_n)$ is an even polynomial in the variables $a_i$ of degree at most $2g$ with rational coefficients, cf.~\cite[Appendix~B]{BR}.
In particular, $I_{g,d;a_1,\dots,a_n}(y)$ is an even polynomial in the variables $a_i$.
Therefore, introducing
\be
u_{k}(z)\= \sum_{a\in\Z}a^k\,\omega_a\,z^a\qquad (k\geq 0),
\ee
(agreeing that $0^0=1$) we can define polynomials $g_k^{poly}(\epsilon,\hbar;\bs u)$ ($k\geq 0$) in the variables $\bs u=(u_0,u_1,\dots)$ with rational coefficients by requiring
\be
\label{eq:diffpolyop}
\NO{g_k^{poly}\bigl(\epsilon,\hbar;\bs u(z)\bigr)} \= g_k(\epsilon;z)\qquad (k\geq 0),
\ee
where $\bs u(z)=(u_0(z),u_1(z),\dots)$ and the normal order defined by Eq.~\eqref{eq:no} is extended by linearity.
Note that $u_k(z)=(z\partial_z)^k u_0(z)$ and that $g_k^{poly}(\epsilon,\hbar;\bs u)$ is even with respect to the grading defining by assigning weight $k$ to $u_k$ for $k\geq 0$ (and weight $0$ to $\hbar$ and $\epsilon$).

\begin{example}
We have (cf.~Example~\ref{ex1})
\begin{align}
g_0^{poly}(\epsilon,\hbar;\bs u)\={}&\frac 12 u_0^2+\frac{\epsilon}{24}u_2-\frac\hbar{24},\quad
\\
g_1^{poly}(\epsilon,\hbar;\bs u)\={}&\frac 16 u_0^3+\frac{\epsilon}{24}u_2u_0+\frac{\epsilon^2}{1152}u_4+\frac{\hbar}{24}(u_2-u_0)+\frac{\epsilon\hbar}{2880}.
\end{align}
{With the substitutions $\epsilon\mapsto-\varepsilon^2$, $\hbar\mapsto\i\hbar$, and $u_s\mapsto({-}\i)^su_s$ ($s\geq 0$) (cf.~Remark~\ref{rem:normalizationdiff}) one obtains the Hamiltonian densities in~\cite[Section~4.1.1]{BR}.}
\end{example}

\section{Partitions and the boson-fermion correspondence}\label{sec:backgroundBF}

In this section, we recall some standard notions in the theory of partitions as well as the so-called infinite-wedge formalism and the boson-fermion correspondence.

\subsection{Partitions}
Denote by $\F=\Z+\tfrac 12$ the set of half-integers, and for a set $X\subseteq\F$ write $X^\pm=\{a\in X \mid \pm a>0\}$.
As already mentioned, we denote by~$\partitions$ the set of all partitions.
For $\lambda\in\partitions$ we denote by $|\lambda|:=\sum_{i}\lambda_i$ the size of $\lambda$ and by $\lambda'$ the transposed partition. Each partition is represented by its Young diagram~$Y_\lambda\mspace{1mu}$, given by
\be
Y_\lambda \= \lbrace (x,y)\in \Z^2\mid\, 1\leq y\leq\lambda_x\rbrace \,.
\ee
We define the \emph{arm-length} $a(\xi)$, \emph{leg-length} $b(\xi)$ and \emph{hook-length} $h(\xi)$ of $\xi=(x,y) \in Y_\lambda$ by
\be
a(\xi) \= \lambda_{x}-y,\quad 
b(\xi)\= \lambda_{y}'-x,\quad
h(\xi)\=a(\xi)+b(\xi)+1.
\ee
Moreover, given partitions $\lambda,\sigma$ with
$Y_\sigma \subset Y_\lambda\mspace{1mu}$, we define a skew Young diagram $\lambda / \sigma$ by removing
the cells of~$Y_\sigma$ from the cells of $Y_\lambda\mspace{1mu}$. We call $\lambda/\sigma$ a \emph{border strip}, if it
is connected (through edges of boxes, not only through vertices) and if it does not
contain a $2\times 2$ block.
The 
\emph{height} $\mathrm{ht}(\gamma)$ of a border strip $\gamma$ is the number of rows it touches minus $1$.
Moreover, the \emph{size} $|\gamma|$ of a border strip $\gamma$ is the number of boxes it consists of.
Note there is a bijection between cells $\xi\in Y_\lambda$ and border strips $\gamma$ contained in $\lambda$, such that $\mathrm{ht}(\gamma) = b(\xi)$ and $|\gamma|=h(\xi)$.

The \emph{Frobenius coordinates} of~$\lambda$ are $(a_1,\dots,a_d\mid b_1,\dots,b_d)$, where $a_i=a(i,i)$, $b_i=b(i,i)$ and $d=\max\lbrace i\geq 1\mid (i,i)\in Y_\lambda\rbrace$.
Clearly, the Frobenius coordinates of~$\lambda$ uniquely identify the partition~$\lambda$.
We also define the \emph{modified Frobenius coordinates} of $\lambda$ as
\be
\label{eq:modFrob}
c_i=a_i+\frac 12\in\mathbb F^+,\quad c_i^*=-b_i-\frac 12\in\mathbb F^-\qquad(i=1,\dots,d).
\ee
We will often make use of the set~$C_\lambda$ of modified Frobenius coordinates, i.e.,
\be
C_\lambda=\lbrace c_i\mid\, i=1,\ldots,d\rbrace
\cup\lbrace c_i^*\mid\, i=1,\ldots,d\rbrace\,.
\ee
Observe that $\sum_{c\in C_\lambda} \mathrm{sgn}(c) \, c = |\lambda|.$ More generally, the shifted symmetric functions, defined by Eq.~\eqref{eq:Qk},  equal
\begin{equation} \label{eq:QkC}
Q_k(\lambda) = \beta_k + \frac{1}{(k-1)!}\sum_{c\in C_\lambda} \mathrm{sgn}(c)\,c^{k-1}.
\end{equation}
This follows from assigning the sequence
\begin{align}\label{eq:Slambda} S_\lambda = \bigl(\lambda_j-j+\tfrac{1}{2} )_{j=1}^{\infty} 
\end{align}
to each $\lambda\in \partitions$ and observing that $S_\lambda$ and $C_\lambda$ are related by
\begin{align}\label{eq:SandC}
\F^+ \cap S_\lambda = C_\lambda^+, \qquad
\F^- \setminus S_\lambda = C_\lambda^-.
\end{align}
Denote by $\mathcal{S}$
the set of all decreasing sequences $S=(s_i)_{i=1}^\infty$ of half-integers for which $s_{i+1}=s_i-1$ for sufficiently large $i$, and by $\mathcal{S}_0\subset \mathcal{S}$ the subset of those sequences~$S$ for which the finite sets $\F^+ \cap S$
and $\F^- \setminus S$ have the same cardinality. Note that the map $\lambda \to S_\lambda$ gives a bijection between $\partitions$ and $\mathcal{S}_0$.

\subsection{Infinite wedge formalism}
Let us recall the formalism of the fermionic Fock space~$\mathsf F$; for more details, we refer to~\cite{AZ,MJD,OkounkovPandharipande,RiosZertuche}.
Let $\mathsf F$ be the $\Q$-vector space consisting of finite linear combinations of formal expressions
\be
\label{eq:basisfermionic}
v_S\=\underline {s_1}\wedge\underline {s_2}\wedge\underline {s_3}\wedge\cdots
\ee
associated to $S=(s_i)_{i=1}^\infty\in \mathcal{S}$. 
We endow $\mathsf F$ with the scalar product $\langle\ ,\ \rangle$ for which the basis $\{v_S \mid S\in \mathcal{S}\}$ is orthonormal.

For $k\in\F$ we consider the operator $\psi_k:\mathsf F\to\mathsf F$ defined by
\be
\psi_k \, v_S\= \underline k\wedge\underline {s_1}\wedge\underline {s_2}\wedge\underline {s_3}\wedge\cdots\,,
\ee
where it is understood that we can put this expression in the canonical form~\eqref{eq:basisfermionic} by using the skew-symmetry of $\wedge$.
Namely, the result is zero whenever $k$ is already in $S$, otherwise, we anti-commute $\underline k$ to the right of the $\underline {s_i}$ such that $s_i>k$.

For $k\in\F$, let $\psi^*_k$ be the adjoint of these operators; their action on the basis vectors~$v_S$ is
\be
\psi_k^*\, v_S\=\begin{cases}
(-1)^{i-1}v_{S\setminus\lbrace s_i\rbrace},& \mbox{if }k=s_i\in S,\\
0,&\mbox{if }k\not\in S.
\end{cases}
\ee

The operators $\psi_k,\psi_k^*$ satisfy the \emph{canonical anticommutation relations}, namely: 
\be
\label{eq:anticomm}
\psi_k\psi_\ell\+\psi_\ell\psi_k\=0,\quad
\psi_k^*\psi_\ell^*\+\psi_\ell^*\psi_k^*\=0,\quad
\psi_k\psi_\ell^*\+\psi_\ell^*\psi_k\=\delta_{k\ell}\quad (k,\ell \in\F).
\ee

We shall be in particular interested in the \emph{charge zero subspace} $\mathsf F_0$, spanned by the $v_S$ with $S\in \mathcal{S}_0$.
We identify $v_\lambda$ with $v_{S_{\lambda}}$ under the bijection between $\partitions$ and $\mathcal{S}_0$. In particular, 
$v_\emptyset$ is associated with the set $\{-\tfrac 12,-\tfrac32,-\tfrac 52,\cdots\}$, and every vector $v_\lambda$ can be written as
\be
\label{eq:vlambda}
v_\lambda\= \sigma_\lambda\,\psi_{c_1}\cdots \psi_{c_d}\psi_{c_d^*}^*\cdots\psi_{c_1^*}^*\,v_\emptyset,\qquad \sigma_\lambda\=(-1)^{\sum_{i=1}^db_i},
\ee
where $c_i$ and $c_i^*$ are the modified Frobenius coordinates of $\lambda$, see Eq.~\eqref{eq:modFrob}, and $b_i$ are (part of) the Frobenius coordinates of $\lambda$.
We also have
\be
\label{eq:ann}
\psi_{-a}\, v_\emptyset \= 0,\quad \psi_{a}^*\,v_\emptyset \= 0\qquad (a\in\F^+).
\ee
Accordingly, for $a\in\F^+$, the operators $\psi_{-a},\psi_a^*$ are termed \emph{annihilation operators} and $\psi_a,\psi_{-a}^*$ are termed \emph{creation operators}.

Next, we recall the \emph{fermionic normal order} $\NO{\text{---}}$, which is obtained by moving all annihilation operators to the right of creation operators (as if they commuted), inserting a minus sign each time we swap two operators.
More explicitly, it is defined recursively by
\begin{align}
\nonumber
\NO{\psi_a\phi_1\cdots\phi_{k}}\=
\begin{cases}
\psi_a\NO{\phi_1\cdots\phi_{k}},&\text{if }a>0, \\
(-1)^k\NO{\phi_1\cdots\phi_{k}}\psi_a,&\text{if }a<0,
\end{cases}
\\
\label{eq:NOF}
\NO{\psi_a^*\phi_1\cdots\phi_{k}}\=
\begin{cases}
(-1)^k\NO{\phi_1\cdots\phi_{k}}\psi_a^*,&\text{if }a>0, \\
\psi_a^*\NO{\phi_1\cdots\phi_{k}},&\text{if }a<0,
\end{cases}
\end{align}
where $\phi_1,\ldots,\phi_{k}\in\lbrace\psi_{i} \mid i \in \F\rbrace\cup \lbrace\psi_{i} ^*\mid i \in \F\rbrace$.
Note in particular that
\be
\label{eq:invarianceNO}
\NO{\phi_{\sigma(1)}\cdots \phi_{\sigma(k)}}\=\mathrm{sgn}(\sigma)\,\NO{\phi_1\cdots\phi_k}\,,
\ee
where $\phi_i$ are as above and $\sigma$ is a permutation of $\lbrace 1,\dots, k\rbrace$.
Finally, the formula (cf.~\cite[Section~2.4]{AZ})
\be
\label{eq:wick}
\phi_1\NO{\phi_2\dots\phi_k}
\,\=\,\NO{\phi_1\phi_2\dots\phi_k}\+\sum_{i=2}^k(-1)^i\langle v_\emptyset,\phi_1\phi_i v_\emptyset\rangle\NO{\phi_1\cdots\phi_{i-1}\phi_{i+1}\cdots\phi_k}
\ee
(where $\phi_i$ are as above) allows one to express normally ordered monomials in $\psi_i,\psi_i^*$ as a linear combination of monomials in $\psi_i,\psi_i^*$, and vice versa, and will be useful below.

\subsection{Boson-fermion correspondence}\label{sec:BF}
Our interest in the space $\mathsf F_0$ is motivated by the \emph{boson-fermion correspondence}, which we recall now.
Let us  introduce the following operators on $\mathsf F$:
\be 
\label{eq:alphaop}
\alpha_n\=\sum_{a\in\F}\NO{\psi_a\psi_{a-n}^*}\qquad (n\in\Z).
\ee
Note that the subspace $\mathsf F_0$ coincides with the kernel of $\alpha_0\mspace{1mu}$.

\begin{theorem}[Boson-fermion correspondence]
\label{thm:bosfer}
Let $\Phi:\mathsf F_0\to\mathsf B$ be the isomorphism of vector spaces defined by $\Phi\,v_\lambda\,=\,s_\lambda$.
Then, for all $n\geq 1$ we have the following identities of operators on $\mathsf F_0$:
\be
\label{eq:bosferop}
\Phi^{-1} \,p_n\,\Phi \= \alpha_n\mspace{1mu},\quad
\Phi^{-1} \,n\pdv{}{p_n}\,\Phi \= \alpha_{-n}\qquad(n\geq 1).
\ee
Conversely, we have the following identity of generating functions of operators on~$\mathsf B$:
\be
\label{eq:bosferop2}
\Phi\,\NO{\psi(z)\psi^*(w)}\,\Phi^{-1}\=\frac 1{z-w}\biggl(\NO{\e^{\varphi(z)-\varphi(w)}}-1\biggr),
\ee
where 
\be\psi(z)=\displaystyle\sum_{a\in\F}\psi_a z^{a-\frac 12},\ \psi^*(z)=\displaystyle\sum_{a\in\F}\psi_a^* z^{-a-\frac 12},\ \text{ and }\ \varphi(z)=\displaystyle\sum_{n\geq 1}\biggl(\frac{p_n}nz^n-\frac{\partial}{\partial p_n}z^{-n}\biggr).
\ee
\end{theorem}

\begin{remark}
\label{rem:isometry}
The map~$\Phi$ is an isometry with respect to the scalar products defined above (for which the $v_\lambda\mspace{1mu}$, respectively the Schur polynomials $s_\lambda\mspace{1mu}$, are orthonormal).
\end{remark}

For the proof, we refer to~\cite{AZ,Kac,MJD}.
{Note that in Eq.~\eqref{eq:bosferop2}, on the left-hand side we have the fermionic normal order whereas on the right-hand side we have the bosonic normal order.}

\section{First order correction to the eigenvalues}\label{sec:proof1}

\subsection{Quadratic and quartic fermionic operators}

Let us introduce the following operators on $\mathsf F_0\mspace{1mu}$, where $a,b,c,d\in\F$:
\be
\Xi_{ab}\=\NO{\psi_a\psi_b^*}\,,\qquad
\Xi_{abcd}\=\NO{\psi_a\psi_b^*\psi_c\psi_d^*}\,.
\ee
Moreover, consider the following variations of the Kronecker delta function:
\be
\label{eq:cocycle}
\eta_{abcd}\=\delta_{ad}^-\delta_{bc}^+-\delta_{ad}^+\delta_{bc}^-\,,\qquad
\delta_{ab}^\pm\=
\begin{cases}
1&\text{if }a=b\text{ and }\pm a>0,\\
0&\text{otherwise}.
\end{cases}
\ee
Note that
\be
\langle v_\emptyset,\psi_a\psi_b^*v_\emptyset\rangle=\delta_{ab}^-\,,\qquad
\langle v_\emptyset,\psi_b^*\psi_a v_\emptyset\rangle = \delta_{ab}^+\,.
\ee

\begin{lemma}
\label{lemma:comm}
For all $a,b,c,d,u,v\in\F$, we have
\begin{align}
\label{eq:lemmacomm2uv}
[\Xi_{ab},\Xi_{uv}]&\=\delta_{bu}\Xi_{av}-\delta_{av}\Xi_{ub}+\eta_{abuv}\,,
\\
\nonumber
[\Xi_{abcd},\Xi_{uv}]&\=
\delta_{du}\Xi_{abcv}-\delta_{cv}\Xi_{abud}+\delta_{bu}\Xi_{avcd}-\delta_{av}\Xi_{ubcd}
\\
\label{eq:lemmacomm4uv}
&\qquad
+\eta_{cduv}\Xi_{ab}
+\eta_{abuv}\Xi_{cd}
-\eta_{cbuv}\Xi_{ad}
-\eta_{aduv}\Xi_{cb}\,.
\end{align}
In particular, when $u=v$, the above formulas simplify to
\begin{align}
\label{eq:lemmacomm2uu}
[\Xi_{ab},\Xi_{uu}]&\=(\delta_{bu}-\delta_{au})\,\Xi_{ab},
\\
\label{eq:lemmacomm4uu}
[\Xi_{abcd},\Xi_{uu}]&\=(\delta_{du}-\delta_{cu}+\delta_{bu}-\delta_{au})\,\Xi_{abcd}\,.
\end{align}
\end{lemma}
\begin{proof}
By Eq.~\eqref{eq:wick} we get
\be
\psi_b^*\NO{\psi_u\psi_v^*}\=\NO{\psi_b^*\psi_u\psi_v^*}+\langle v_\emptyset,\psi_ b^*\psi_u v_\emptyset\rangle \psi_v^*\=\NO{\psi_b^*\psi_u\psi_v^*}+\delta_{bu}^+\psi_v^*
\ee
and, again by an application of the same formula, we get
\begin{align}
\nonumber
\psi_a\NO{\psi_b^*\psi_u\psi_v^*}\={}&\NO{\psi_a\psi_b^*\psi_u\psi_v^*}+\langle v_\emptyset ,\psi_a\psi_b^*v_\emptyset\rangle\NO{\psi_u\psi_v ^*}+
\langle v_\emptyset ,\psi_a\psi_v^*v_\emptyset\rangle\NO{\psi_b^*\psi_u}
\\
\={}&
\NO{\psi_a\psi_b^*\psi_u\psi_v^*}+\delta_{ab}^-\NO{\psi_u\psi_v ^*}-\delta_{av}^-\NO{\psi_u\psi_b^*}
\end{align}
Noting that $\Xi_{ab}=\psi_a\psi_b^*-\delta^-_{ab}$ and combining these two identities we conclude that
\be
\label{eq:4to2}
\Xi_{ab}\Xi_{uv}\=\Xi_{abuv}\meno\delta_{av}^-\Xi_{ub}\+\delta_{bu}^+\Xi_{av}\+\delta_{av}^-\delta_{bu}^+.
\ee
(This is essentially a special case of Wick's theorem.)
Subtract from Eq.~\eqref{eq:4to2} the same relation after the substitutions $a\leftrightarrow u$ and $b\leftrightarrow v$ to obtain Eq.~\eqref{eq:lemmacomm2uv}.
Here, we use~$\Xi_{abuv}=\Xi_{uvab}\mspace{1mu}$, cf. Eq.~\eqref{eq:invarianceNO}.
Next, using Eq.~\eqref{eq:4to2}, we have
\begin{align}
\nonumber
[\Xi_{abcd},\Xi_{uv}]&=[\Xi_{ab}\Xi_{cd},\Xi_{uv}]+\delta_{ad}^-[\Xi_{cb},\Xi_{uv}]-\delta_{bc}^-[\Xi_{ad},\Xi_{uv}]
\\
&=\Xi_{ab}[\Xi_{cd},\Xi_{uv}]+[\Xi_{ab},\Xi_{uv}]\Xi_{cd}+\delta_{ad}^-[\Xi_{cb},\Xi_{uv}]-\delta_{bc}^-[\Xi_{ad},\Xi_{uv}]\,.
\end{align}
After some algebraic simplifications and using~Eqs.~\eqref{eq:lemmacomm2uv} and~\eqref{eq:4to2}, we obtain Eq.~\eqref{eq:lemmacomm4uv}.
Finally, Eqs.~\eqref{eq:lemmacomm2uu} and~\eqref{eq:lemmacomm4uu} follow from~Eqs.~\eqref{eq:lemmacomm2uv} and~\eqref{eq:lemmacomm4uv}, respectively, as ${\eta_{abuu}=0}$.
\end{proof}

\subsection{Fermionic expressions for KdV Hamiltonian operators}

Let us denote $\wh g^{[j]}_k(z)=\rho_0\bigl(g_k^{[j]}(z)\bigr)\in(\mathrm{End}\,\mathsf B)\llbracket z^{\pm 1}\rrbracket$, where $\rho_0$ is defined by Eq.~\eqref{eq:repB} and $g_k^{[j]}(z)$ is defined by Eq.~\eqref{eq:decomp}.
Recall that $\Phi$ denotes the isomorphism in the boson-fermion correspondence, cf.~Theorem~\ref{thm:bosfer}.
The following result was essentially proved by Rossi \cite[Theorem~5.3]{Rossi}.
We reprove it here for the sake of completeness using a different strategy based on the recursion relation recalled in Section~\ref{sec:backgroundqDR}.
We will later employ the same strategy to study the first order term $\Phi^{-1}\,\wh g^{[1]}_k(z)\,\Phi$ as well.

\begin{proposition}
\label{prop:G0fermion}
For all $k\geq -1$, we have
\be
\label{eq:G0fermion}
\Phi^{-1}\,\wh g^{[0]}_k(z)\,\Phi\=\beta_{k+2}\+\frac 1{(k+1)!}\sum_{a,b\in\F}\left(\frac{a+b}2\right)^{\! k+1}\Xi_{ab}\,z^{{a-b}},
\ee
where the numbers $\beta_k$ are defined by Eq.~\eqref{eq:betagen}.
\end{proposition}

Before giving the proof of this proposition, we define
\be
\label{eq:compGhat}
\wh G_k^{[j]}\=\rho_0(G_k^{[j]}),
\ee
where $\rho_0$ is defined by Eq.~\eqref{eq:repB} and $G_k^{[j]}$ is defined by Eq.~\eqref{eq:decomp}.
Then, we have the following lemma.

\begin{lemma}\label{lemma}
We have
\begin{align}
\label{eq:LEMMA1}
\Phi^{-1}\,\wh G_1^{[0]}\,\Phi\={}&\sum_{a\in\F}\frac{a^2}2\,\Xi_{aa}\,,
\\
\label{eq:LEMMA2}
\quad\Phi^{-1} \, \wh G_1^{[1]}\, \Phi\={}&\sum_{\substack{a,b,c,d\in\F,\\ a+c=b+d}}(a-b)^2\,\Xi_{abcd}\+\sum_{a\in\F}\frac 23a\biggl(a^2-\frac 14\biggr)\,\Xi_{aa}\+ \frac 1{120}\,.
\end{align}
\end{lemma}
\begin{proof}
By applying $z^3\partial_z\partial_w$ and setting $w=z$ in Eq.~\eqref{eq:bosferop2} we obtain the following identity of generating functions of operators on~$\mathsf B$:
\be
\label{eq:verylast}
-\frac 13\NO{\bigl(z\partial_z\varphi(z)\bigr)^3}+\frac 16z^3\partial_z^3\varphi(z) \= \Phi\,\NO{z^3\pa_z\psi(z)\pa_z\psi^*(z)}\,\Phi^{-1}\,.
\ee
The coefficient of $z^0$ in the left-hand side of Eq.~\eqref{eq:verylast} is $-2\,\wh G_1^{[0]}$, cf. Eq.~\eqref{eq:g1example}.
On the other hand, the right-hand side can be expanded using
\be
\NO{z^3\pa_z\psi(z)\pa_z\psi^*(z)}\=\sum_{a,b\in\F}\biggl(a-\frac 12\biggr)\biggl(-b-\frac 12\biggr)\,\NO{\psi_a\psi_b^*}\,z^{a-b}.
\ee
Hence, the coefficient of $z^0$ on the right-hand side  of Eq.~\eqref{eq:verylast} is
\be
-\Phi\sum_{a\in\F}\biggl(a^2-\frac 14\biggr)\,\NO{\psi_a\psi_a^*}\Phi^{-1}\=-\Phi\sum_{a\in\F}a^2\,\Xi_{aa}\Phi^{-1},
\ee
where we use that $\alpha_0=\sum_{a\in\F}\Xi_{aa}$ vanishes on $\mathsf F_0\mspace{1mu}$. Therefore, Eq.~\eqref{eq:LEMMA1} is proved.

Next, to prove Eq.~\eqref{eq:LEMMA2}, we note that by Eq.~\eqref{eq:G1rec} we have
\be
G_1^{[1]}\=\sum_{a\in\Z}a^2\NO{\omega_{a}\omega_{-a}}\+\frac{\hbar}{120}\,.
\ee
Hence,
\be
\label{eq:G1Explicit}
\wh G_1^{[1]}\=2\sum_{m\geq 1}m^3p_m\frac{\partial}{\partial p_m}\+\frac 1{120}\,,
\ee
and so, using Eq.~\eqref{eq:bosferop},
\be
\label{eq:firsteq}
\Phi^{-1}\,\wh G_1^{[1]}\,\Phi\meno\frac 1{120}\=
2\sum_{m\geq 1}m^2\,\alpha_m\,\alpha_{-m}
\=2\sum_{m\geq 1}m^2\sum_{a,b\in\F}\Xi_{a,a-m}\,\Xi_{b,b+m}\,.
\ee
By an application of Eq.~\eqref{eq:4to2}, we rewrite this expression as
\begin{align}
\nonumber
&2\sum_{m\geq 1}m^2\sum_{a,b\in\F}\left(\Xi_{a,a-m,b,b+m}-\delta_{a,b+m}^-\Xi_{b,a-m}+\delta^+_{a-m,b}\Xi_{a,b+m}+\delta_{a,b+m}^-\delta_{a-m,b}^+\right)
\\
&\quad = 2\sum_{m\geq 1}m^2\left(\sum_{a,b\in\F}\Xi_{a,a-m,b,b+m}\right)+2\sum_{a\in\F}\Xi_{aa}\left(\delta_{a>0}\sum_{m=1}^{a-\frac 12}m^2-\delta_{a<0}\sum_{m=1}^{-a-\frac 12}m^2\right)\raisebox{-10pt}{.}
\end{align}
We finally use the identity $\Xi_{abcd}=\Xi_{cdab}$ in the first term and 
\be 
\delta_{a>0}\sum_{m=1}^{a-\frac 12}m^2-\delta_{a<0}\sum_{m=1}^{-a-\frac 12}m^2\=\tfrac 13a(a^2-\tfrac 14) 
\qquad (a\in\F)
\ee
in the second term to rewrite this expression as in Eq.~\eqref{eq:LEMMA2}.
\end{proof}

\begin{proof}[Proof of Proposition~\ref{prop:G0fermion}]
As recalled in Section~\ref{sec:rec}, the recursion given in Eq.~\eqref{eq:BRrec} uniquely determines the $g_k(\epsilon;z)$ up to an additive constant (depending on $\epsilon,\hbar$ only).
Indeed, even though the coefficient of $z^0$ is killed by the operator $z\partial_z$ appearing in the recursion, the polynomiality of the coefficients $I_{g,d;a_1,\dots,a_n}(y)$ in $a_1,\dots,a_n$ (see Eq.~\eqref{eq:gk}) allows us to recover this coefficient up to a constant depending on $\epsilon,\hbar$ only. Moreover, the recursion itself is insensitive to the actual value of the constant depending on $\epsilon,\hbar$ only.
Equivalently, the recursion~\eqref{eq:BRrec} uniquely determines the $g_k(\epsilon;z)$ up an additive constant (depending on $\epsilon,\hbar$ only) provided that $g_k(\epsilon;z)$ are of the form of Eq.~\eqref{eq:diffpolyop} for some differential polynomials.
Therefore, denoting by $j_k(z)$ the right-hand side in Eq.~\eqref{eq:G0fermion}, the proof consists of showing three statements about $j_k(z)$: 
\begin{itemize}\itemsep3pt
\item[(1)] that the recursion
\be
\label{eq:tobeproveng0}
{-}(k+2)\,z\partial_z\, j_{k+1}(z)\=\bigl[j_k(z),\Phi^{-1}\wh G_1^{[0]}\Phi\bigr]
\ee
holds for all $k\geq -1$,
\item[(2)] that $\Phi j_k(z) \Phi^{-1}$ (for all $k\geq -1$) are of the form of Eq.~\eqref{eq:diffpolyop} for appropriate differential polynomials,
\item[(3)] that the constant $\beta_{k+2}$ in $j_k(z)$ coincides with the constant in $\wh g_k^{[0]}(z)$.
\end{itemize}

We start with the first of these statements.
The left-hand side of Eq.~\eqref{eq:tobeproveng0} equals
\be
\frac{(k+2)}{(k+2)!}\sum_{a,b\in\F}(b-a)\biggl(\frac{a+b}2\biggr)^{k+2}\Xi_{ab}\,z^{{a-b}},
\ee
while the right-hand side is, using Eqs.~\eqref{eq:lemmacomm2uu} and~\eqref{eq:LEMMA1},
\begin{align}
\nonumber
&\frac 1{(k+1)!}\sum_{u,a,b\in\F}\frac {u^2}2\biggl(\frac{a+b}2\biggr)^{k+1}[\Xi_{ab},\Xi_{uu}]\,z^{{a-b}}
\\
\nonumber
&\qquad=\frac 1{(k+1)!}\sum_{u,a,b\in\F}(\delta_{bu}-\delta_{au})\frac {u^2}2\biggl(\frac{a+b}2\biggr)^{k+1}\Xi_{ab}\,z^{{a-b}}
\\
&\qquad=\frac 1{(k+1)!}\sum_{a,b\in\F}\frac{(b^2-a^2)}{2}\biggl(\frac{a+b}2\biggr)^{k+1}\Xi_{ab}\,z^{{a-b}}.
\end{align}
This proves Eq.~\eqref{eq:tobeproveng0}.

For the second statement, we use the notations of Eq.~\eqref{eq:bosferop2} to write
\be
j_k(z) \= \beta_{k+2}+\frac z{(k+1)!}\biggl(\frac{z\partial_z-w\partial_w}{2}\biggr)^{k+1}\NO{\psi(z)\psi^*(w)}\, \bigr|_{w=z}\,.
\ee
Hence, using Eq.~\eqref{eq:bosferop2}, we have
\be
\Phi j_k(z) \Phi^{-1} = \beta_{k+2}+\frac z{(k+1)!}\biggl(\frac{z\partial_z-w\partial_w}{2}\biggr)^{k+1}\frac{\NO{\e^{\varphi(z)-\varphi(w)}}-1}{z-w}\, \biggr|_{w=z}\,,
\ee
which, since the right-hand side is a (normally ordered) polynomial expression in $(z\partial_z)^l\varphi(z)$ for $l\geq 1$, is of the form of Eq.~\eqref{eq:diffpolyop}.

Finally, the constant in $\wh g_k^{[0]}(z)$ is equal to the constant in $\wh G_k^{[0]}(\epsilon)$.
By looking at~\eqref{eq:spectrum} with $\lambda=\emptyset$, using $r_\emptyset=1$ and Eq.~\eqref{eq:E0}, we see that this constant is $Q_{k+2}(\emptyset)=\beta_{k+2}\mspace{1mu}$.
\end{proof}

Next, we establish a formula for $\Phi^{-1}\,\wh g^{[1]}_k(z)\,\Phi$ in terms of quartic fermionic operators. 
This formula involves certain polynomials, which we introduce first.

\begin{lemma}\label{lem:ABpols}
There exist unique polynomials $A_k=A_k(a,b,c,d)$ and $B_k=B_k(a,b)$ satisfying the  initial values
\be
\label{eq:rrecic}
A_{-1}(a,b,c,d)\=0\,,\qquad B_{-1}(a,b)\=0
\ee
and recursions {\upshape ($k\geq 0$)}
\begin{align}
\label{eq:polrec1}
(k+2)\,(a-b+c-d)\,A_{k}(a,b,c,d)
&\=
\frac{a^2-b^2+c^2-d^2}2 A_{k-1}(a,b,c,d) \meno P_{k}(a,b,c,d)\,,
\\
\label{eq:polrec2}
(k+2)\,(a-b)\,B_{k}(a,b)
&\=
\frac{a^2-b^2}2 B_{k-1}(a,b)\meno R_{k}(a,b)
\meno T_{k}(a,b)\,,
\end{align}
where
\begin{align}
\nonumber
P_k(a,b,c,d)&\=\frac 1{k!\,2^k}\biggl((a-b)^2\bigl( (-a+b+c+d)^{k} - (a-b+c+d)^{k} \bigr) 
\\
\label{eq:Ppoly}
&\qquad\qquad\quad
+(c-d)^2\bigl((a+b-c+d)^{k} - (a+b+c-d)^{k}\bigr) \biggr)\,,
\\
\label{eq:Spoly}
R_k(a,b)&\=[\zeta^k]\,\e^{\zeta(a-b)/2}\bigl((b-\partial_\zeta)^2-(a-b)^2\bigr)\frac{1-\e^{-(a-b)\zeta}}{\sinh(\zeta/2)}\,,
\\
\label{eq:Tpoly}
T_k(a,b)&\=\frac 2{3\,k!}\left(b(b^2-\tfrac 14)-a(a^2-\tfrac 14)\right)\Bigl(\frac {a+b}2\Bigr)^k\,.
\end{align}
Moreover, for all $k\geq 1$, the polynomial~$A_k$ has degree $k+1$, and, for all $k\geq 0$, polynomial~$B_k$ has degree $k+2$, and we have
\be
\label{eq:reductionAgeneral}
A_k(a,b,c,a+c-b) = \frac 1{(k+1)!}(a-b)\frac{a^{k+1} - b^{k+1} +c^{k+1} - (a+c-b)^{k+1}}{b-c}
\ee
as well as
\be
\label{eq:reductionB}
B_k(a,a) \= \frac{k(k+3)}{(k+2)!}a^{k+2}\meno\frac 1{12\,k!}a^{k}\+ 2\beta_{k+1}a\meno 2(k+1)\beta_{k+2} \mspace{1mu},
\ee
where the~$\beta_k$ are defined by Eq.~\eqref{eq:betagen}.
\end{lemma}
In Eq.~\eqref{eq:Spoly}, the notation $[\zeta^k]f(\zeta)$ stands for the coefficient of $\zeta^k$ in the power series $f(\zeta)\in\mathbb Q\llbracket\zeta\rrbracket$.
Note that~Eq.~\eqref{eq:reductionAgeneral} implies
\be
\label{eq:reductionA}
A_k(a,a,c,c)\=0,\qquad A_k(a,b,b,a)\=\frac{(a-b)(a^k-b^k)}{k!}.
\ee

\begin{proof}
It is not hard to see that $R_k$ and $T_k$ are divisible by $a-b$. Hence, the right-hand side of Eq.~\eqref{eq:polrec2} is divisible by $a-b$.
So, there exists a unique polynomial solution to the recursion in Eq.~\eqref{eq:polrec2} with the initial value given in Eq.~\eqref{eq:rrecic}.

In order to show the existence of polynomial solutions for the recursion of the $A_k\mspace{1mu}$, we note that the recursion is equivalent to the following ordinary differential equation
\be
\label{eq:recAODE}
(a-b+c-d)\,(y\partial_y+2)\,\mathbf A(y)\=\frac{a^2-b^2+c^2-d^2}2\,y\,\mathbf A(y)\meno\mathbf P(y)
\ee
for the formal power series
\begin{align}
\mathbf A(y)\={}&\sum_{k\geq 0}A_k(a,b,c,d)\,y^k
\\
\mathbf P(y) \={}& \sum_{k\geq 0}P_k(a,b,c,d)\,y^k\\
\={}& \nonumber
(a-b)^2\bigl(\e^{\frac y2(-a+b+c+d)}-\e^{\frac y2(a-b+c+d)}\bigr)+(c-d)^2\bigl(\e^{\frac y2(a+b-c+d)}-\e^{\frac y2(a+b+c-d)}\bigr)\,.
\end{align}
The general solution to~Eq.~\eqref{eq:recAODE} is
\be
\mathbf A_0(y)\+C\frac{1}{y^2}\exp\biggl(y\frac{a^2-b^2+c^2-d^2}{2 (a-b+c-d)}\biggr)
\ee
for an arbitrary constant $C$ and a special solution $\mathbf A_0(y)$ given explicitly by
\begin{align}
\nonumber
\mathbf A_0(y)&\= \frac{\e^{\frac y2 (-a+b+c+d)} \bigl( (a-b) (a-d)y+t\bigr) \,+\, \e^{\frac y2 (a+b+c-d)} \bigl( (a-d) (c-d)y-t\bigr)}{(a-d)^2\,y^2}
\\
\label{eq:A0}
&\qquad-\frac {\e^{\frac y2 (a-b+c+d)} \bigl( (a-b) (b-c)y+t\bigr)\,+\,\e^{\frac y2 (a+b-c+d)}  \bigl( (b-c) (c-d)y-t\bigr)}{(b-c)^2\,y^2}\,,
\end{align}
where $t=a-b+c-d$.
The only solution which is regular at $y=0$ is obtained for $C=0$.
Hence, $\mathbf A(y)$ coincides with the Taylor series of $\mathbf A_0(y)$ at $y=0$.
By looking at the recursion, one deduces that $A_k$ is a rational function of $a,b,c,d$ for which the denominator equals a power of $a-b+c-d$. 
Since $\mathbf A_0(y)$ has a well-defined limit 
\be
\lim_{d\to a-b+c}\mathbf A_0(y) \= (a-b)\frac{\left(\e^{y a}-\e^{y b}+\e^{c y}-\e^{y (a-b+c)}\right)}{y (b-c)}\,,
\ee
we infer that $A_k(a,b,c,d)$ is actually a polynomial in $a,b,c,d$.
Moreover, the last equation also proves Eq.~\eqref{eq:reductionAgeneral}.

The reduction~in Eq.~\eqref{eq:reductionB} follows by making the recursion for $B_k(a,a)$ explicit:
\be
\label{eq:polrec2Baa}
(k+3)B_{k+1}(a,a)\=a\, B_k(a,a)+\wt R_{k+1}(a)+\wt T_{k+1}(a)\,,
\ee
where
\begin{align}
\nonumber
\wt R_k(a)&\=\lim_{b\to a}\frac{R_k(a,b)}{b-a}\\
\nonumber &\=[\zeta^k]\biggl((a-\partial_\zeta)^2\frac{\zeta}{\sinh(\zeta/2)}\biggr)\\
&\=-2\beta_ka^2\+ 4(k+1)\beta_{k+1}a\meno 2(k+1)(k+2)\beta_{k+2}\,,
\\
\wt T_k(a)&\=\lim_{b\to a}\frac{T_k(a,b)}{b-a}\=\frac 2{k!}a^k\bigl(a^2-\tfrac 1{12}\bigr).
\end{align}
It is then readily checked that the right-hand side of Eq.~(\ref{eq:reductionB}) satisfies the same recursion in Eq.~\eqref{eq:polrec2Baa} and the same initial data, thus proving Eq.~(\ref{eq:reductionB}).

Finally, since the degrees of $P_k\mspace{1mu}$, $R_k$ and $T_k$ do not exceed $k+2$, $k+3$ and $k+3$ respectively, it is clear that the degree of $A_k$ cannot exceed $k+1$ and the degree of $B_k$ cannot exceed $k+2$.
That the degrees equal these bounds (except for $A_0$ which vanishes, as it can be checked by a direct computation) follows from the explicit reductions in Eqs.~\eqref{eq:reductionAgeneral} and~\eqref{eq:reductionB}.
\end{proof}

\begin{proposition}
\label{prop:G1fermion}
For all $k\geq -1$, we have
\be
\label{eq:G1fermion}
\Phi^{-1}\,\wh g^{[1]}_k(z)\,\Phi\=\!\sum_{a,b,c,d\in\F}\!A_k(a,b,c,d)\,\Xi_{abcd}\,z^{{a-b+c-d}}\+\sum_{a,b\in\F}B_k(a,b)\,\Xi_{ab}\,z^{{a-b}}\+\gamma_k\,,
\ee
where $\gamma_k\in\Q$ and $A_k$ and $B_k$ are the polynomials given in Lemma~\ref{lem:ABpols}.
\end{proposition}

\begin{remark}
\label{rem:tk}
The numbers $\gamma_k$ are not computed in this proposition.
It will be shown below, in the proof of Theorem~\ref{thm:1}, that $\gamma_k=2\beta_2\beta_{k+1}+k(k+3)\beta_{k+3}\mspace{1mu}$, where $\beta_k$ is defined by Eq.~\eqref{eq:betagen}.
\end{remark}

\begin{proof}
The strategy is the same as in the proof of Proposition~\ref{prop:G0fermion}, with the difference that we do not need compute explicitly the constants $\gamma_k\mspace{1mu}$.
Therefore, as explained in that proof, we only need to show that, denoting by $h_k(z)$ the right-hand side in Eq.~\eqref{eq:G1fermion}, we have
\be
\label{eq:rrec}
-(k+3)\,z\partial_z\, h_{k+1}(z)\=\bigl[h_k(z),\Phi^{-1}\wh G_1^{[0]}\Phi\bigr]\+\bigl[\Phi^{-1}\wh g_k^{[0]}(z)\Phi,\Phi^{-1}\wh G_1^{[1]}\Phi\bigr]
\ee
and that~$\Phi h_k(z)\Phi^{-1}$ is of the form of Eq.~\eqref{eq:diffpolyop} for an appropriate differential polynomial.
We start with the first statement.
The left-hand side of Eq.~\eqref{eq:rrec} is
\begin{align}
\nonumber
&(k+3)\,\biggl(\sum_{a,b,c,d\in\F}(d-c+b-a)\,A_{k+1}(a,b,c,d)\,\Xi_{abcd}\,z^{{a-b+c-d}}
\\&\qquad\qquad\qquad\qquad\qquad\qquad\qquad\+\sum_{a,b\in\F}(b-a)\,B_{k+1}(a,b)\,\Xi_{ab}\,z^{{a-b}}\biggr).
\end{align}
The first term on the right-hand side of Eq.~\eqref{eq:rrec} is, using Lemmas~\ref{lemma:comm} and~\ref{lemma},
\begin{align}
\nonumber
&\bigl[h_k(z),\Phi^{-1}\wh G_1^{[0]}\Phi\bigr]
\\
\nonumber
&\=\!
\sum_{\ell,a,b,c,d\in\F}\!\frac{\ell^2}2A_k(a,b,c,d)\,[\Xi_{abcd},\Xi_{\ell\ell}]\,z^{{a-b+c-d}}
\+\sum_{\ell,a,b\in\F}\frac{\ell^2}2B_k(a,b)\,[\Xi_{ab},\Xi_{\ell\ell}]\,z^{{a-b}}
\\
&\=\!\sum_{a,b,c,d\in\F}\!\frac{d^2-c^2+b^2-a^2}2A_k(a,b,c,d)\,\Xi_{abcd}\,z^{{a-b+c-d}}\+\sum_{a,b\in\F}\frac{b^2-a^2}2B_k(a,b)\,\Xi_{ab}\,z^{{a-b}}.
\end{align}
The second term on the right-hand side of Eq.~\eqref{eq:rrec} is, once again using Lemma~\ref{lemma},
\begin{align}
\nonumber
\bigl[\Phi^{-1}\wh g_k^{[0]}(z)\Phi,\Phi^{-1}\wh G_1^{[1]}\Phi\bigr]
&\=
\frac 1{(k+1)!}\sum_{\substack{a,b,c,d,u,v\in\F,\\ a+c=b+d}}(a-b)^2\biggl(\frac{u+v}{2}\biggr)^{k+1}[\Xi_{uv},\Xi_{abcd}]\,z^{{u-v}}
\\
\label{eq:twolines}
&\quad
\+\frac 1{(k+1)!}\sum_{a,u,v\in\F}\frac 23a\biggl(a^2-\frac 14\biggr)\biggl(\frac{u+v}{2}\biggr)^{\! k+1}[\Xi_{uv},\Xi_{aa}]\,z^{{u-v}}.
\end{align}
It follows from Lemma~\ref{lemma:comm} that this expression contains terms which are either quartic or quadratic in the $\psi_i,\psi_i^*$, which we now consider separately.

The quartic terms come from the first line only and they are
\begin{align}
\nonumber
&\frac{-1}{(k+1)!}\sum_{\substack{a,b,c,d,u,v\in\F,\\ a+c=b+d}}(a-b)^2\biggl(\frac{u+v}{2}\biggr)^{\! k+1}(\delta_{du}\Xi_{abcv}-\delta_{cv}\Xi_{abud}+\delta_{bu}\Xi_{avcd}-\delta_{av}\Xi_{ubcd})\,z^{{u-v}}
\\
\nonumber
&\qquad\=
\meno\frac 1{(k+1)!}\sum_{a,b,c,v\in\F}(a-b)^2\biggl(\frac{a-b+c+v}{2}\biggr)^{\!  k+1}\,\Xi_{abcv}\,z^{{a-b+c-v}}
\\
\nonumber
&\qquad\qquad
\+\frac 1{(k+1)!}\sum_{a,b,d,u\in\F}(a-b)^2\biggl(\frac{u+b+d-a}{2}\biggr)^{\!  k+1}\,\Xi_{abud}\,z^{{a-b+u-d}}
\\
\nonumber
&\qquad\qquad
\meno\frac 1{(k+1)!}\sum_{a,c,d,v\in\F}(c-d)^2\biggl(\frac{a+c-d+v}{2}\biggr)^{\!  k+1}\,\Xi_{avcd}\,z^{{a-v+c-d}}
\\
\nonumber
&\qquad\qquad
\+\frac 1{(k+1)!}\sum_{b,c,d,u\in\F}(c-d)^2\biggl(\frac{u+b+d-c}{2}\biggr)^{\!  k+1}\,\Xi_{ubcd}\,z^{{u-b+c-d}}
\\
&\qquad\=\sum_{a,b,c,d\in\F}P_{k+1}(a,b,c,d)\,\Xi_{abcd}\,z^{{a-b+c-d}}\,,
\end{align}
with $P_k(a,b,c,d)$ as in Eq.~\eqref{eq:Ppoly}.
Therefore, using Eq.~\eqref{eq:polrec1}, we get that the coefficients of Eq.~\eqref{eq:rrec} match; here we note that the initial conditions~\eqref{eq:rrecic} agree with Eq.~\eqref{eq:initialdatumBRrec}.

The quadratic terms instead come from both lines on the right-hand side of Eq.~\eqref{eq:twolines}.
Those coming from the first line are
\begin{align}
\frac {-1}{(k+1)!}\sum_{\substack{a,b,c,d,u,v\in\F,\\ a+c=b+d}}(a-b)^2\biggl(\frac{u+v}{2}\biggr)^{\!  k+1}(\eta_{cduv}\Xi_{ab}
+\eta_{abuv}\Xi_{cd}
-\eta_{cbuv}\Xi_{ad}
-\eta_{aduv}\Xi_{cb})\,z^{{u-v}}\,.
\end{align}
By renaming summation indexes in the second and fourth terms by the transformations $a\leftrightarrow c$ and $b\leftrightarrow d$, we see that this expression equals
\begin{align}
\nonumber
\meno&\frac 2{(k+1)!}\sum_{\substack{a,b,c,d,u,v\in\F,\\ a+c=b+d}}\!\!\!(a-b)^2\biggl(\frac{u+v}{2}\biggr)^{\!  k+1}+\bigl(\eta_{cduv}\Xi_{ab}-\eta_{cbuv}\Xi_{ad}\bigr)\,z^{{u-v}}
\\
\nonumber
& =\meno
\frac 2{(k+1)!}\sum_{\substack{a,b,c,d,u,v\in\F,\\ a+c=b+d}}\!\!\!\bigl((a-b)^2-(a-d)^2\bigr)\biggl(\frac{u+v}{2}\biggr)^{\! k+1}\eta_{cduv}\Xi_{ab}\,z^{{u-v}}
\\
\nonumber
&=\meno
\frac 2{(k+1)!}\left(\sum_{\substack{a,b,c\in\F, \\ c<0,\ c+a-b>0}}\!\!\!\bigl((a-b)^2-(b-c)^2\bigr)\biggl(\frac{a+2c-b}{2}\biggr)^{\!  k+1}\Xi_{ab}\,z^{{a-b}}
\right.
\\
&
\nonumber
\qquad\qquad\qquad\qquad
\left.
-\!\!\!\sum_{\substack{a,b,c\in\F, \\ c>0,\ c+a-b<0}}\!\!\!\bigl((a-b)^2-(b-c)^2\bigr)\biggl(\frac{a+2c-b}{2}\biggr)^{\!  k+1}\Xi_{ab}\,z^{{a-b}}\right)
\\
&\=\sum_{a,b\in\F}\widehat R_{k+1}(a,b)\,\Xi_{ab}\,z^{{a-b}},
\end{align}
where $\widehat R_k(a,b)$ is defined as
\be
\label{eq:thelastexpression}
\widehat R_k(a,b)\=\frac 2{k!}\Biggl(\sum_{\substack{c\in\F^-, \\ c+a-b>0}}\!\!\!\bigl((b-c)^2-(a-b)^2\bigr)\bigl(\tfrac{a+2c-b}{2}\bigr)^{k}
-\!\!\!\sum_{\substack{c\in\F^+, \\c+a-b<0}}\!\!\!\bigl((b-c)^2-(a-b)^2\bigr)\bigl(\tfrac{a+2c-b}{2}\bigr)^{k}\Biggr).
\ee
Let us show that $\widehat R_k(a,b)=R_k(a,b)$, where $R_k(a,b)$ is as in Eq.~\eqref{eq:Spoly}:
\begin{align}
\nonumber
\widehat R_k(a,b) &\= 2[\zeta^k]\,\e^{\zeta(a-b)/2}\bigl((b-\partial_\zeta)^2-(a-b)^2\bigr)\biggl[\sum_{\substack{c\in\F^-, \\ c+a-b>0}}\e^{c\,\zeta}-\sum_{\substack{c\in\F^+, \\ c+a-b<0}}\e^{c\,\zeta}\biggr] 
\\
&\=[\zeta^k]\e^{\zeta(a-b)/2}\bigl((b-\partial_\zeta)^2-(a-b)^2\bigr)\frac{1-\e^{-(a-b)\zeta}}{\sinh(\zeta/2)}\=R_k(a,b)\,.
\end{align}
The quadratic terms coming from the second line of Eq.~\eqref{eq:twolines} are computed by Eq.~\eqref{eq:lemmacomm2uu} as
\begin{multline}
\frac 1{(k+1)!}\sum_{u,v\in\F}\frac 23\biggl(v\biggl(v^2-\frac 14\biggr)-u\biggl(u^2-\frac 14\biggr)\!\biggr)\,\biggl(\frac{u+v}{2}\biggr)^{\!k+1}\Xi_{uv}\,z^{{u-v}} \\
\=
\sum_{a,b\in\F}T_{k+1}(a,b)\,\Xi_{ab}\,z^{{a-b}}\,.
\end{multline}
Therefore, using Eq.~\eqref{eq:polrec2}, also the terms of the form $\Xi_{ab}z^{{a-b}}$ in Eq.~\eqref{eq:rrec} coincide; here we note that the initial conditions in Eq.~\eqref{eq:rrecic} agree with Eq.~\eqref{eq:initialdatumBRrec}.

It only remains to argue that $\Phi h_k(z)\Phi^{-1}$ is of the form of Eq.~\eqref{eq:diffpolyop} for appropriate differential polynomials.
To this end, we use the notations in~Theorem~\ref{thm:bosfer} to write
\begin{align}
\nonumber
h_k(z) &= z^2 \, A_k\bigl(D_{z_1},-D_{w_1},D_{z_2},-D_{w_2}\bigr)\,\NO{\psi(z_1)\psi^*(w_1)\psi(z_2)\psi^*(w_2)}\Bigr|_{\substack{z_i=z,\\ w_i=z}}
\\
&\quad\+z\,B_k(D_{z},-D_{w})\,\NO{\psi(z)\psi^*(w)}\bigr|_{w=z}\+\gamma_k
\end{align}
where $D_z=z\partial_z+\frac 12$.
To complete the proof, we claim that the following identity holds:
\begin{align}
\nonumber
&\Phi\,\NO{\psi(z_1)\psi^*(w_1)\psi(z_2)\psi^*(w_2)}\,\Phi^{-1}
\\
\nonumber
&\quad
\=\frac{(z_1-z_2)(w_1-w_2)}{(z_1-w_2)(w_1-z_2)(z_1-w_1)(z_2-w_2)}\NO{\e^{\varphi(z_1)-\varphi(w_1)+\varphi(z_2)-\varphi(w_2)}}
\\ \label{eq:quarticfermionboson} &\quad\qquad
\meno\frac{\NO{\e^{\varphi(z_1)-\varphi(w_1)}}
+\NO{\e^{\varphi(z_2)-\varphi(w_2)}}
-1}{(z_1-w_1)(z_2-w_2)}
\meno
\frac{\NO{\e^{\varphi(z_1)-\varphi(w_2)}}
+\NO{\e^{\varphi(z_2)-\varphi(w_1)}}
-1}{(z_1-w_2)(w_1-z_2)}\,.
\end{align}
Note that the right-hand side in Eq.~\eqref{eq:quarticfermionboson} is regular along all diagonals $z_1=z_2$, $w_1=w_2$, $z_i=w_j$ ($i,j=1,2$).
Therefore, assuming this formula and recalling Eq.~\eqref{eq:bosferop2}, we can infer that $\Phi h_k(z)\Phi^{-1}$ are indeed of the form of Eq.~\eqref{eq:diffpolyop} for appropriate differential polynomials exactly as we did in the proof of Proposition~\ref{prop:G0fermion}.

It only remains to prove Eq.~\eqref{eq:quarticfermionboson}.
To this end, we first note the identity
\begin{align}
\nonumber
&\NO{\e^{\varphi(z_1)-\varphi(w_1)}}\,\NO{\e^{\varphi(z_2)-\varphi(w_2)}} 
\\
\nonumber
&\= 
\e^{\sum_{n\geq 1}(z_1^n-w_1^n)\frac{p_n}n}\e^{-\sum_{n\geq 1}(z_1^{-n}-w_1^{-n})\partial_{p_n}}
\e^{\sum_{n\geq 1}(z_2^n-w_2^n)\frac{p_n}n}\e^{-\sum_{n\geq 1}(z_2^{-n}-w_2^{-n})\partial_{p_n}}
\\
\nonumber
&\=\e^{\sum_{n\geq 1}(z_1^n-w_1^n+z_2^n-w^n_2)\frac{p_n}n}\e^{-\sum_{n\geq 1}(z_1^{-n}-w_1^{-n}+z_2^{-n}-w_2^{-n})\partial_{p_n}}
\e^{-\sum_{n\geq 1}\frac 1n(z_1^{-n}-w_1^{-n})(z_2^n-w_2^n)}
\\
&\=\frac{(w_1-w_2) (z_1-z_2)}{(w_1-z_2) (z_1-w_2)}\,\NO{\e^{\varphi(z_1)-\varphi(w_1)+\varphi(z_2)-\varphi(w_2)}}
\end{align}
and then we similarly use Eq.~\eqref{eq:4to2} to get the identity of generating series
\begin{align}
\nonumber
\NO{\psi(z_1)\psi^*(w_1)}\,\NO{\psi(z_2)\psi^*(w_2)}&\=\NO{\psi(z_1)\psi^*(w_1)\psi(z_2)\psi^*(w_2)}-\frac 1{z_1-w_2}\NO{\psi(z_2)\psi^*(w_1)}
\\
&\qquad +\frac 1{w_1-z_2}\NO{\psi(z_2)\psi^*(w_1)}+\frac 1{(z_1-w_2)(w_1-z_2)}.
\end{align}
Using the last two identities and Eq.~\eqref{eq:bosferop2}, one obtains the claimed Eq.~\eqref{eq:quarticfermionboson}.
\end{proof}

\begin{corollary}
\label{corollary:G1fermion}
We have
\be
\label{eq:G1opfermion}
\Phi^{-1}\,\wh G^{[1]}_k\,\Phi\=\sum_{a,b,c\in\F}A_k(a,b,c,a+c-b)\,\Xi_{a,b,c,a+c-b}\+\sum_{a\in\F}B_k(a,a)\,\Xi_{aa}\+\gamma_k\,,
\ee
where the polynomials $A_k$ and $B_k$ are defined in Lemma~\ref{lem:ABpols}, cf.~Eqs.~\eqref{eq:reductionAgeneral} and~\eqref{eq:reductionB}.
\end{corollary}

\subsection{Proof of Theorem~\ref{thm:1}}

We are now ready for the proof of Theorem~\ref{thm:1}.
We first have the following simple lemma.

\begin{lemma}
\label{lemma:simple}
We have
\be
E^{[1]}_k(\lambda;0)\=\frac 1{24}\langle s_\lambda,\wh G^{[1]}_ks_\lambda\rangle \=\frac 1{24}
\langle v_\lambda,\Phi^{-1}\wh G^{[1]}_k\Phi v_\lambda\rangle\,.
\ee
\end{lemma}
\begin{proof}
The second equality follows from Theorem~\ref{thm:bosfer}.
The first one follows by considering the linear terms in~$\epsilon$ in Eq.~\eqref{eq:spectrum} (setting $c=0$); here one has to note that 
\be
\langle s_\lambda,\wh G_k^{[0]}r_\lambda^{[1]}\rangle = \langle \wh G_k^{[0]}s_\lambda,r_\lambda^{[1]}\rangle = 0
\ee
where we make use of the fact that $\wh G_k^{[0]}$ is diagonal on the Schur basis (hence symmetric with respect to the scalar product $\langle\ ,\ \rangle$) and of Eq.~\eqref{eq:gauge}.
\end{proof}

According to Lemma~\ref{lemma:simple} and Corollary~\ref{corollary:G1fermion}, we have 
\be
\label{eq:start}
24 \mspace{1mu} E^{[1]}_k(\lambda;0)\=\!\!\sum_{a,b,c\in\F}\!\!A_k(a,b,c,a+c-b)\langle v_\lambda,\Xi_{a,b,c,a+c-b}v_\lambda\rangle\+\sum_{a\in\F}B_k(a,a)\langle v_\lambda,\Xi_{aa}v_\lambda\rangle\+\gamma_k\,.
\ee
All the ingredients needed to simplify this expression are given in the following proposition.
Recall the constant terms~$\beta_k$ of the shifted symmetric functions~$Q_k$ defined by Eq.~\eqref{eq:betagen}.

\begin{proposition}\label{prop:final}
Let $a,b,c,d\in \F$ and $k\geq 0$.
The following are true:
\begin{enumerate}[{\upshape(i)}]
\item 
$\label{it:quadraticexpectation}
\displaystyle\langle v_\lambda,\Xi_{aa}v_\lambda\rangle \=\sum_{c\in C_\lambda}{\rm sgn}(c)\,\delta_{a,c}\mspace{1mu}.
$

\item\label{it:nonzero}
The scalar product $\langle v_\lambda,\Xi_{abcd}v_\lambda\rangle$ is nonzero only if ($a=b$ and $c=d$), or ($a=d$ and $b=c$).

\item \label{it:quarticexpectation}
$\displaystyle
\langle v_\lambda,\Xi_{abba}v_\lambda\rangle \= -\sum_{c,c'\in C_\lambda}{\rm sgn}(c\,c')\,\delta_{a,c}\,\delta_{b,c'}\qquad (a\,\not=\,b)\mspace{1mu}.
$\vspace{-3pt}
\end{enumerate}
\end{proposition}
\begin{proof}
(\ref{it:quadraticexpectation})
Let us first show that
\be
\label{eq:Eaa}
\Xi_{aa}v_\lambda\=\Biggl(\,\sum_{c\in C_\lambda}{\rm sgn}(c)\,\delta_{ac}\Biggr)v_\lambda\,.
\ee
As a special case of Eq.~\eqref{eq:NOF}, note that
\be
\Xi_{aa}\=\NO{\psi_a\psi_a^*}\=
\begin{cases}
\psi_a\psi_a^*,& \mbox{if }a>0,\\
-\psi_a^*\psi_a,& \mbox{if }a<0.
\end{cases}
\ee
That is, when $a\in \F^+$ we have $\Xi_{aa}v_\lambda=\sigma_\lambda\psi_a\psi_a^*\psi_{c_1}\cdots\psi_{c_d}\psi_{c_d^*}^*\cdots\psi_{c_1^*}^*v_\emptyset$ (cf. Eq.~\eqref{eq:vlambda}). Using the commutation relations~\eqref{eq:anticomm} and that $\psi_a^*v_\emptyset=0$, we infer that this expression is nonzero only if $a=c_i$ for some $i$, in which case $\Xi_{aa}v_\lambda=v_\lambda\mspace{1mu}$.
Similarly, when $a\in \F^-$ one obtains $\Xi_{aa}v_\lambda=-\sigma_\lambda\psi_a^*\psi_a\psi_{c_1}\cdots\psi_{c_d}\psi_{c_d^*}^*\cdots\psi_{c_1^*}^*v_\emptyset$. Since $\psi_av_\emptyset=0$, the latter expression is nonzero only if $a=c_i^*$ for some $i$, in which case $\Xi_{aa}v_\lambda=-v_\lambda$.
Eq.~\eqref{eq:Eaa} follows and~(\ref{it:quadraticexpectation}) is an immediate consequence.

(\ref{it:nonzero}) According to Eq.~\eqref{eq:4to2}, we have
\be
\label{eq:4to2exp}
\langle v_\lambda,\Xi_{abcd}v_\lambda\rangle \=
\langle v_\lambda,\Xi_{ab}\Xi_{cd}v_\lambda\rangle 
+\delta_{ad}^-\langle v_\lambda,\Xi_{cb}v_\lambda\rangle 
-\delta_{bc}^+\langle v_\lambda,\Xi_{ad}v_\lambda\rangle
-\delta_{ad}^-\delta_{bc}^+\,.
\ee
Since $\langle v_\lambda,\Xi_{uv}v_\lambda\rangle $ can be nonzero only if $u=v$, the last three terms on the right-hand side can be nonzero if and only if $a=d$ and $b=c$. Since $(\Xi_{ab})^*=\Xi_{ba}$, the first term on the right-hand side is $\langle \Xi_{ba}v_\lambda,\Xi_{cd} v_\lambda\rangle$, which can be nonzero only if $a=d$ and $b=c$ or if $a=b$ and $c=d$.

(\ref{it:quarticexpectation})
Since $a\not=b$, it can be checked from Eq.~\eqref{eq:NOF} that $\Xi_{abba}=-\Xi_{aa}\Xi_{bb}$ and~(\ref{it:quarticexpectation}) follows from Eq.~\eqref{eq:Eaa}.
\end{proof}

We are ready to prove our first main result.
\begin{proof}[Proof of Theorem~\ref{thm:1}]
We start from Eq.~\eqref{eq:start} and use Proposition~\ref{prop:final}, as well as Eqs.~\eqref{eq:reductionB} and~\eqref{eq:reductionA}, to obtain the following expression for $24\, E^{[1]}_k(\lambda;0)$:
\begin{multline}
\sum_{c,c'\in C_\lambda}\frac{c^kc'+c\,{c'}^k-c^{k+1}-{c'}^{k+1}}{k!}{\rm sgn}(c\,c') \+
\\
\+\sum_{c\in C_\lambda}\biggl(\frac{k(k+3)}{(k+2)!}c^{k+2}\meno \frac 1{12\,k!}c^{k}\+ 2\beta_{k+1}c\meno 2(k+1)\beta_{k+2}\biggr){\rm sgn}(c)\+\gamma_k\,.
\end{multline}
Observe, using Eq.~\eqref{eq:QkC} and $\sum_{c\in C_\lambda}{\rm sgn}(c)\=0$, that the first line is $2(Q_2-\beta_2)(Q_{k+1}-\beta_{k+1})$ and that the second line is
\be
k(k+3)(Q_{k+3}-\beta_{k+3})\+ 2\beta_2(Q_{k+1}-\beta_{k+1})\+2\beta_{k+1}(Q_2-\beta_2)+\gamma_k\mspace{1mu}.
\ee
Hence, 
\be
24\, E^{[1]}_k(\lambda;0)\=2Q_2Q_{k+1}\+k(k+3)Q_{k+3}\+\gamma_k\meno 2\beta_2\beta_{k+1}\meno k(k+3)\beta_{k+3}\mspace{1mu}.
\ee
Using Theorems~\ref{thm:BO} and~\ref{thm:IR}, we deduce that the constant term $\gamma_k-2\beta_2\beta_{k+1}-k(k+3)\beta_{k+3}$ must vanish, which concludes the proof (and yields the explicit value for~$\gamma_k$ anticipated in Remark~\ref{rem:tk}).
This proves Eq.~\eqref{eq:E1simple}, and hence Theorem~\ref{thm:1}.
\end{proof}

\subsection{Proof of Corollary~\ref{corollary:doubleHodge}}\label{sec:proofcorollary}
\begin{proof}
Recalling Eqs.~\eqref{eq:gk} and \eqref{eq:I}, the coefficient of $\epsilon$ in the constant term of $G_k(\epsilon)$ is
\be
I_{g,k;\emptyset}\=\int_{{\rm DR}_g(0)}\psi_1^k\lambda_1\=\,(-1)^g\int_{\overline{\mathcal M}_{g,1}}\lambda_1\lambda_g\psi^k_1\mspace{1mu},
\ee
where we use the fact that ${\rm DR}_g(0)$ is Poincar\'e dual to $(-1)^g\lambda_g$, cf.~\cite[Section~0.5.3]{JPPZ}.
This intersection number is nonzero only if $3g-2=1+g+k$, i.e., only if $k=2g-3$.
On the other hand, using Eq.~\eqref{eq:spectrum1} with $\lambda=\emptyset$, the fact that $r_\emptyset=1$, and Theorem~\ref{thm:1}, the same term is computed as the value of Eq.~\eqref{eq:E1simple} at $\lambda=\emptyset$, for which we finally use the identity $Q_k(\emptyset)=\beta_k\mspace{1mu}$ (cf. Eq.~\eqref{eq:Qk}).
\end{proof}

\section{First order correction to the eigenvector}
\label{introrefsecthm2}
\subsection{A metric on partitions}

Given $\lambda,\mu\in \partitions$, we define the \emph{Hamming distance} between $\lambda$ and $\mu$ as
\begin{equation}\label{eq:Hamming}
d(\lambda,\mu)\=\frac{1}{2}\bigl|C_\lambda\, \Delta\, C_\mu\bigr|
\end{equation}
where $X\,\Delta\,Y = (X\cup Y)\setminus (X\cap Y)$ is the symmetric difference of sets $X,Y$.

\begin{lemma}\ \vspace{-3pt}
\label{lemma:Hamming}
    \begin{enumerate}[{\upshape(i)}]\itemsep1pt
    \item\label{metric:i} The Hamming distance~$d$ is an integer-valued metric on $\partitions$. 
    \item\label{metric:ii} We have $d(\lambda,\mu)=\frac{1}{2}\bigl|S_\lambda\, \Delta\, S_\mu\bigr|$, where $S_\lambda$ is given by Eq.~\eqref{eq:Slambda}.
    \item\label{point} If $|\lambda|=|\mu|$, then $d(\lambda,\mu)\neq 1$. 
    \item\label{metric:iii} For all $k\geq 0$, if $d(\lambda,\mu)=2k$, then $|S_\lambda\setminus S_\mu|\=2k\=|S_\mu\setminus S_\lambda|$.
    \end{enumerate}
\end{lemma}
\begin{proof}

It is well-known that the symmetric difference defines a metric.
Hence, for (\ref{metric:i}) it remains to show that $d$ only takes integer values. Since $C_\lambda \,\Delta\, C_\mu = (C_\lambda \setminus C_\mu)\,\sqcup\, (C_\mu \setminus C_\lambda)$ and $\sum_{c\in C_\lambda} \mathrm{sgn}(c) = 0$ we obtain that
\be\label{eq:Q1}
\sum_{c\in C_\lambda\setminus C_\mu}\mathrm{sgn}(c)  = \sum_{c\in C_\mu\setminus C_\lambda} \mathrm{sgn}(c).
\ee
In particular, $|C_\lambda \setminus C_\mu|$ and $|C_\mu \setminus C_\lambda|$ have the same parity, proving that $d$ is integer-valued.
 
For~(\ref{metric:ii}), recall from Eq.~\eqref{eq:SandC} that $C_\lambda^+=S_\lambda^+$ and $C_\lambda^-=\mathbb F^-\setminus S_\lambda^-$. Hence,
\be
C_\lambda\,\Delta\, C_\mu = (C_\lambda^+\,\Delta\, C_\mu^+)\cup (C_\lambda^-\,\Delta\, C_\mu^-)=(S_\lambda^+\,\Delta\, S_\mu^+)\cup (S_\lambda^-\,\Delta\, S_\mu^-)=S_\lambda\,\Delta\, S_\mu,
\ee
where we use that the symmetric difference is invariant under taking complements, i.e.,
\be
(Z\setminus X)\,\Delta\, (Z\setminus Y)=X\Delta Y \qquad (\text{for all sets $X,Y\subseteq Z$}).
\ee

Next, for (\ref{point}) recall that $|\lambda|=\sum_{c\in C_\lambda}|c|$ for all $\lambda\in\partitions$. Hence,  if $|\lambda|=|\mu|$ we have
\be
\sum_{c\in C_\lambda\setminus C_\mu} |c|  = \sum_{c\in C_\mu\setminus C_\lambda} |c|.
\ee
If $d(\lambda,\mu)=1$, i.e., $|C_\mu\, \Delta\, C_\lambda|=2$, this implies that $|C_\lambda\setminus C_\mu|=1=|C_\mu\setminus C_\lambda|$, and hence because of Eq.~\eqref{eq:Q1} that $C_\lambda\setminus C_\mu=C_\mu\setminus C_\lambda\mspace{1mu}$, contradiction.

Finally, for any pair of partitions $\lambda,\mu$ let us denote $b^\pm_{\lambda\mu}=|S_\lambda^\pm\setminus S_\mu^\pm|$. If $d(\lambda,\mu)=2k$ we must have
\be
\label{eq:lemma1}
b_{\lambda\mu}^++b_{\mu\lambda}^++b_{\lambda\mu}^-+b_{\mu\lambda}^-=4k.
\ee
We have $b_{\lambda\mu}^+=|C_\lambda^+\setminus C_\mu^+|$ and $b_{\lambda\mu}^-=|C_\mu^-\setminus C_\lambda^-|$ (note the order of $\lambda$ and $\mu$). Hence, condition~\eqref{eq:Q1} implies
\be
b_{\lambda\mu}^+-b_{\mu\lambda}^-=
b_{\mu\lambda}^+-b_{\lambda\mu}^-\mspace{1mu}.
\ee
Therefore, combining the last identity with Eq.~\eqref{eq:lemma1}, we get $b_{\lambda\mu}^++b_{\lambda\mu}^-=2k=b_{\mu\lambda}^++b_{\mu\lambda}^-\mspace{1mu}$, and~(\ref{metric:iii}) is proved.
\end{proof}

Let us introduce the following {\it neighborhood} of $\lambda$:
\be
\label{eq:defU}
\mathcal{U}(\lambda)\=\bigl\lbrace\mu\in \partitions\ \big|\  |\mu|=|\lambda| \text{ and } d(\lambda,\mu)=2\bigr\rbrace.
\ee
According to the last point in the previous lemma, for $\mu \in \mathcal{U}(\lambda)$, there exist $1\leq a<b$ and $1\leq a'<b'$ such that
\be
\label{eq:notationU}
S_\lambda\setminus S_\mu = \lbrace \lambda_a-a+\tfrac 12,\,\lambda_b-b+\tfrac 12\rbrace\qquad
S_\mu\setminus S_\lambda = \lbrace \mu_{a'}-{a'}+\tfrac 12,\,\mu_{b'}-{b'}+\tfrac 12\rbrace.
\ee
Note that, since $|\lambda|=|\mu|$, for all $K$ large enough
\be
\sum_{i=1}^{K}(\lambda_i-i)=\sum_{i=1}^{K}(\mu_i-i)
\ee
(in fact, $K\geq \ell(\lambda),\ell(\mu)$ suffices) and so
\be
\label{eq:remarkforsimplification}
\lambda_a-a+\lambda_b-b=\mu_{a'}-{a'}+\mu_{b'}-{b'}.
\ee

Before proceeding, it is convenient to illustrate our arguments with one example.
\begin{example}
Let $\lambda=(7,2,1)$ and $\mu=(4,2,2,2)$, which satisfy $|\lambda|=|\mu|=10$ and $d(\lambda,\mu)=2$.
Consider the sequences $S_\lambda$ and $S_\mu$ depicted as Maya diagram~{\upshape \cite{MJD,RiosZertuche}:}
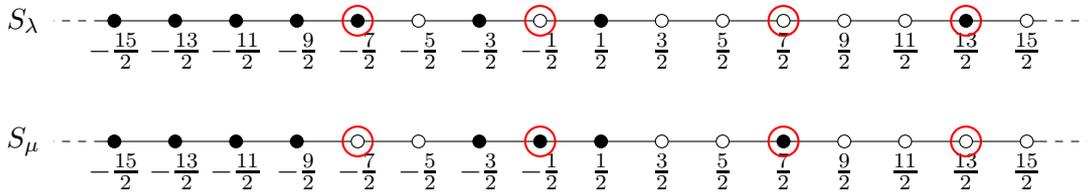
\[
\begin{tikzpicture}
\begin{scope}[scale=.8]

\draw (7.7,0) -- (-7.7,0);
\draw (7.7,0)[dashed] -- (8.5,0);
\draw (-7.7,0)[dashed] -- (-8.5,0);

\draw[fill=white](-1/2,0) circle (3 pt);
\draw[fill=black](-3/2,0) circle (3 pt);
\draw[fill=white](-5/2,0) circle (3 pt);
\draw[fill=black](-7/2,0) circle (3 pt);
\draw[fill=black](-9/2,0) circle (3 pt);
\draw[fill=black](-11/2,0) circle (3 pt);
\draw[fill=black](-13/2,0) circle (3 pt);
\draw[fill=black](-15/2,0) circle (3 pt);
\draw[fill=black](1/2,0) circle (3 pt);
\draw[fill=white](3/2,0) circle (3 pt);
\draw[fill=white](5/2,0) circle (3 pt);
\draw[fill=white](7/2,0) circle (3 pt);
\draw[fill=white](9/2,0) circle (3 pt);
\draw[fill=white](11/2,0) circle (3 pt);
\draw[fill=black](13/2,0) circle (3 pt);
\draw[fill=white](15/2,0) circle (3 pt);

\node at (-1/2,-1/2) {$-\tfrac 12$};
\node at (-3/2,-1/2) {$-\tfrac 32$};
\node at (-5/2,-1/2) {$-\tfrac 52$};
\node at (-7/2,-1/2) {$-\tfrac 72$};
\node at (-9/2,-1/2) {$-\tfrac 92$};
\node at (-11/2,-1/2) {$-\tfrac {11}2$};
\node at (-13/2,-1/2) {$-\tfrac {13}2$};
\node at (-15/2,-1/2) {$-\tfrac {15}2$};
\node at (1/2,-1/2) {$\tfrac 12$};
\node at (3/2,-1/2) {$\tfrac 32$};
\node at (5/2,-1/2) {$\tfrac 52$};
\node at (7/2,-1/2) {$\tfrac 72$};
\node at (9/2,-1/2) {$\tfrac 92$};
\node at (11/2,-1/2) {$\tfrac {11}2$};
\node at (13/2,-1/2) {$\tfrac {13}2$};
\node at (15/2,-1/2) {$\tfrac {15}2$};

\node at (-9,0) {$S_\lambda$};
\begin{scope}[shift={(0,-2)}]

\draw (7.7,0) -- (-7.7,0);
\draw (7.7,0)[dashed] -- (8.5,0);
\draw (-7.7,0)[dashed] -- (-8.5,0);

\draw[fill=black](-1/2,0) circle (3 pt);
\draw[fill=black](-3/2,0) circle (3 pt);
\draw[fill=white](-5/2,0) circle (3 pt);
\draw[fill=white](-7/2,0) circle (3 pt);
\draw[fill=black](-9/2,0) circle (3 pt);
\draw[fill=black](-11/2,0) circle (3 pt);
\draw[fill=black](-13/2,0) circle (3 pt);
\draw[fill=black](-15/2,0) circle (3 pt);
\draw[fill=black](1/2,0) circle (3 pt);
\draw[fill=white](3/2,0) circle (3 pt);
\draw[fill=white](5/2,0) circle (3 pt);
\draw[fill=black](7/2,0) circle (3 pt);
\draw[fill=white](9/2,0) circle (3 pt);
\draw[fill=white](11/2,0) circle (3 pt);
\draw[fill=white](13/2,0) circle (3 pt);
\draw[fill=white](15/2,0) circle (3 pt);

\node at (-1/2,-1/2) {$-\tfrac 12$};
\node at (-3/2,-1/2) {$-\tfrac 32$};
\node at (-5/2,-1/2) {$-\tfrac 52$};
\node at (-7/2,-1/2) {$-\tfrac 72$};
\node at (-9/2,-1/2) {$-\tfrac 92$};
\node at (-11/2,-1/2) {$-\tfrac {11}2$};
\node at (-13/2,-1/2) {$-\tfrac {13}2$};
\node at (-15/2,-1/2) {$-\tfrac {15}2$};
\node at (1/2,-1/2) {$\tfrac 12$};
\node at (3/2,-1/2) {$\tfrac 32$};
\node at (5/2,-1/2) {$\tfrac 52$};
\node at (7/2,-1/2) {$\tfrac 72$};
\node at (9/2,-1/2) {$\tfrac 92$};
\node at (11/2,-1/2) {$\tfrac {11}2$};
\node at (13/2,-1/2) {$\tfrac {13}2$};
\node at (15/2,-1/2) {$\tfrac {15}2$};

\node at (-9,0) {$S_\mu$};

\end{scope}

\draw[red,thick](-1/2,0) circle (7 pt);
\draw[red,thick](-1/2,-2) circle (7 pt);
\draw[red,thick](-7/2,0) circle (7 pt);
\draw[red,thick](-7/2,-2) circle (7 pt);
\draw[red,thick](7/2,0) circle (7 pt);
\draw[red,thick](7/2,-2) circle (7 pt);
\draw[red,thick](13/2,0) circle (7 pt);
\draw[red,thick](13/2,-2) circle (7 pt);

\end{scope}
\end{tikzpicture}
\]
Here, the black dots correspond to elements in~$S_\lambda$ and~$S_\mu\mspace{1mu}$, and we indicated in red the four places at which $S_\lambda$ and $S_\mu$ differ;
note that, in this case, the indices in Eq.~\eqref{eq:remarkforsimplification} are given by $a=1,$ $b=4,$ $a'=1,$ $b'=3$.

Next, recall that in Maya diagrams, removing a border strip is represented by swapping a black dot with any white dot to its left.
Therefore we see that there are exactly two ways to remove a border strip from $\lambda$ and one from $\mu$ such that the resulting partitions coincide.
Namely, we either swap the white and black dots at $\tfrac 72$ and $\tfrac {13}2$ in~$S_\lambda$ and at $-\tfrac 72$ and $-\tfrac 12$ in~$S_\mu\mspace{1mu}$, or we swap white and black dots at $-\tfrac 12$ and $\tfrac {13}2$ in~$S_\lambda$ and at $-\tfrac 72$ and $\tfrac 72$ in~$S_\mu\mspace{1mu}$.
These border strips are visualized in the Young diagrams as follows 
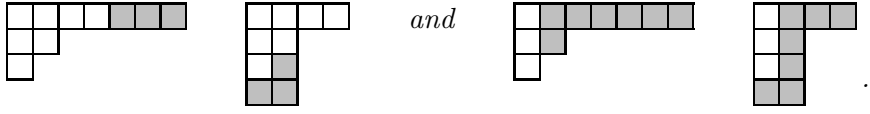
\[
\ytableaushort{{}{}{}{}{*(lightgray)}{*(lightgray)}{*(lightgray)},{}{},{}}\qquad
\ytableaushort{{}{}{}{},{}{},{}{*(lightgray)},{*(lightgray)}{*(lightgray)}} \qquad \text{and} \qquad 
\ytableaushort{{}{*(lightgray)}{*(lightgray)}{*(lightgray)}{*(lightgray)}{*(lightgray)}{*(lightgray)},{}{*(lightgray)},{}}\qquad
\ytableaushort{{}{*(lightgray)}{*(lightgray)}{*(lightgray)},{}{*(lightgray)},{}{*(lightgray)},{*(lightgray)}{*(lightgray)}}\,\raisebox{-22pt}{.}
\]
\end{example}

The general situation is completely parallel, as shown in the following lemma.
\begin{lemma}
\label{lemma:borderstrips}
Let $\mu\in\mathcal U(\lambda)$.
There exist exactly two border strips $\gamma_1,\gamma_2$ in $\lambda$ and two border strips $\gamma_1',\gamma_2'$ in $\mu$ such that $\lambda\setminus \gamma_i = \mu\setminus \gamma_i '$ ($i=1,2$).
These border strips satisfy
\begin{enumerate}[{\upshape (i)}]
\item\label{it:border1}
$\displaystyle
\gamma_1\subset\gamma_2\,,\ \gamma_1'\subset \gamma_2'\,,$
\item \label{it:border2}
$\displaystyle
|\gamma_1|=|\gamma_1'|=|\lambda_a-a-\mu_{a'}+{a'}|,\ 
|\gamma_2|=|\gamma_2'|=|\lambda_a-a-\mu_{b'}+{b'}|
$
\item If $ \lambda_a-a>\mu_{a'}-{a'}$, then 
\be
\mathrm{ht}(\gamma_1)=a'-a,\
\mathrm{ht}(\gamma_1')=b-b',\
\mathrm{ht}(\gamma_2)=b'-a-1,\
\mathrm{ht}(\gamma_2')=b-a'
\ee
and if $\lambda_a-a<\mu_{a'}-{a'}$, then
\be
\mathrm{ht}(\gamma_1)=b'-b,\
\mathrm{ht}(\gamma_1')=a-a',\
\mathrm{ht}(\gamma_2)=b'-a,\
\mathrm{ht}(\gamma_2')=b-a'-1.
\ee
\end{enumerate}
\end{lemma}
\begin{proof}
We can assume without loss of generality that $\lambda_a-a>\mu_{a'}-a'$ (otherwise we just interchange the roles of $\lambda$ and $\mu$). Since $a$ and $a'$ are minimal such that $\lambda_a-a$ and $\mu_{a'}-a'$ are not elements of $S_\mu$ and $S_\lambda$ respectively, this implies $a'\geq a$. Using Eq.~\eqref{eq:remarkforsimplification} one concludes $a\leq a'<b'\leq b$. 
There are exactly two ways to remove border strips from $\lambda$ and $\mu$ such that the resulting partitions coincide: we either swap white and black dots at $\mu_{a'}-a'+\tfrac 12$ and $\lambda_a-a+\tfrac 12$ in $S_\lambda$ and those at $\lambda_b-b+\tfrac 12$ and $\mu_{b'}-b'+\tfrac 12$ in $S_\mu$ (corresponding to border strips $\gamma_1,\gamma_1'$ respectively), or we swap white and black dots at $\mu_{b'}-b'+\tfrac 12$ and $\lambda_a-a+\tfrac 12$ in $S_\lambda$ and those at $\lambda_b-b+\tfrac 12$ and $\mu_{a'}-a'+\tfrac 12$ in $S_\mu$ (corresponding to border strips $\gamma_2,\gamma_2'$ respectively).
The properties~(\ref{it:border1}) and~(\ref{it:border2}) are clear by construction.
The last one follows from the fact that the height of a border strip in a partition~$\nu$ is read off the Maya diagram as follows: the border strip itself is a pair of white and black dots, the white one to the left of the right one, and the height is the number of black dots (i.e., of elements of $S_\nu$) in between them (see, e.g.,~\cite{RiosZertuche}).
\end{proof}

We now compute the scalar products relevant to the proof of Theorem~\ref{thm:2}.

\begin{lemma}
\label{lemma:scalarproduct}
Let $\lambda,\mu\in\partitions$ with $\lambda\not=\mu$.
\begin{enumerate}[{\upshape (i)}]
\item\label{it:scalarproduct-i} If $\mu\not\in\mathcal U(\lambda)$, then $\langle s_\mu, \wh G_k^{[1]} s_\lambda\rangle=0$ for all $k\geq 0$.
\item\label{it:scalarproduct-ii} If $\mu\in\mathcal U(\lambda)$, then (with the notations of the previous lemma)
\begin{align}
\nonumber
\langle s_\mu,\wh G_1^{[1]}s_\lambda\rangle\={}&2\,(-1)^{\mathrm{ht}(\gamma_1)+\mathrm{ht}(\gamma_1')}|\gamma_1|^2+(-1)^{\mathrm{ht}(\gamma_2)+\mathrm{ht}(\gamma_2')}|\gamma_2|^2
\\
\={}&
    2\,(-1)^{a+a'+b+b'}\,\bigl((\lambda_a-a-\mu_{a'}+a')^2-(\lambda_a-a-\mu_{b'}+b')^2\bigr)\,.
\end{align}
\end{enumerate}
\end{lemma}
\begin{remark}
{Directly by Eq.~\eqref{eq:G1Explicit} we see that the left-hand side in the last identity involves twice the quantity $\sum_{m\geq 1}m^3\,\bigl\langle s_\mu,p_m\pdv{s_\lambda}{p_m}\bigr\rangle$.}
It is straightforward to adapt the proof below to show that, more generally, we have
\begin{align}
    \nonumber
\sum_{m\geq 1}m^{n+1}\,\Bigl\langle s_\mu,p_m\pdv{s_\lambda}{p_m}\Bigr\rangle\={}&
(-1)^{\mathrm{ht}(\gamma_1)+\mathrm{ht}(\gamma_1')}|\gamma_1|^n+(-1)^{\mathrm{ht}(\gamma_2)+\mathrm{ht}(\gamma_2')}|\gamma_2|^n
\\
\={}&
    (-1)^{a+a'+b+b'}\,\bigl((\lambda_a-a-\mu_{a'}+a')^n-(\lambda_a-a-\mu_{b'}+b')^n\bigr)
\end{align}
for any integer $n\geq 0$ and partitions $\lambda,\mu$ satisfying $\mu\in\mathcal U(\lambda)$.
\end{remark}
\begin{proof}
Recalling Corollary~\ref{corollary:G1fermion} we obtain that, when $\lambda\not=\mu$, $\langle s_\mu, \wh G_k^{[1]} s_\lambda\rangle$ is a linear combination of $\langle v_\mu,\Xi_{a,b,c,a+c-b}v_\lambda\rangle$ and $\langle v_\mu,\Xi_{aa}v_\lambda\rangle$, for $a,b,c\in\mathbb F$.
That the latter terms do not contribute is clear from the fact that $\Xi_{aa}$ is diagonal on $v_\lambda$, cf. Eq.~\eqref{eq:Eaa}.
Moreover, $\Xi_{abcd}v_\lambda$ is a linear combination of $v_\mu$ such that at most four modified Frobenius coordinates of $\mu$ are different from those of $\lambda$, i.e., such that $d(\mu,\lambda)\leq 2$. Since $|\lambda|=|\mu|$ and $\lambda\not=\mu$ also these other terms do not contribute, cf.~point~(\ref{point}) in~Lemma~\ref{lemma:Hamming}).
It remains to prove the last assertion. To this end, we use Eq.~\eqref{eq:firsteq} to write
\be
\langle s_\mu,\wh G_1^{[1]}s_\lambda\rangle\=2\,\sum_{n\geq 1} n^2\langle v_\lambda,\alpha_n\alpha_{-n}v_\mu\rangle\qquad (\lambda\not=\mu)
\ee
where $\alpha_n$ are the operators defined by Eq.~\eqref{eq:alphaop}.
It is well known (see, for example,~\cite{RiosZertuche}) that
\be
\alpha_{-n}v_\lambda \= \sum_{\substack{\gamma\in \mathrm{BS}(\lambda,n)}}(-1)^{\mathrm{ht}(\gamma)}v_{\lambda\setminus\gamma}\qquad (n\geq 1)
\ee
where~$\mathrm{BS}(\lambda,n)$ is the set of border strips of size $n$ in a partition $\lambda$.
Since the operators $\alpha_n$ and $\alpha_{-n}$ are mutually adjoint, cf.~\cite{RiosZertuche}, we get
\be
\langle s_\mu,\wh G_1^{[1]}s_\lambda\rangle\=2\,\sum_{n\geq 1}n^2\langle \alpha_{-n} v_\mu,\alpha_{-n}v_\lambda\rangle \= 2\,\sum_{n\geq 1}n^2\!\!\!\sum_{\substack{\gamma\in \mathrm{BS}(\lambda,n)\\ \gamma'\in \mathrm{BS}(\mu,n)}}\!\!\!(-1)^{\mathrm{ht}(\gamma)+\mathrm{ht}(\gamma')}\,\delta_{\lambda\setminus\gamma,\mu\setminus\gamma'}\,.
\ee
Finally, we can use Lemma~\ref{lemma:borderstrips} to complete the proof.
\end{proof}

\subsection{Proof of Theorem~\ref{thm:2}}\label{sec:proof2}

With the notations of Eqs.~\eqref{eq:spectrum} and~\eqref{eq:spectrum1}, we define coefficients $c(\lambda,\mu)$, depending on pairs of partitions $\lambda,\mu$, by
\be
\label{eq:defa}
r^{[1]}_\lambda(\bs p) \= \frac 1{24}\sum_{\mu\in\partitions}c(\lambda,\mu)\,s_\mu(\bs p) \,.
\ee
By Eq.~\eqref{eq:gauge}, $c(\lambda,\lambda)=0$ for all~$\lambda\in\partitions$.
Obviously, since $r_\lambda(\epsilon)\in \mathsf B_{|\lambda|}\llbracket\epsilon\rrbracket$, one has $c(\lambda,\mu)=0$ unless $|\lambda|=|\mu|$.

\begin{lemma}\label{lemma:simple2}
For all $\lambda,\mu\in\partitions$ we have
\be
c(\lambda,\mu)\=\begin{cases}
\displaystyle\frac{\langle s_\mu, \wh G_1^{[1]}s_\lambda\rangle}{Q_3(\lambda)-Q_3(\mu)} & \text{if }Q_3(\lambda)\not=Q_3(\mu),
\\
0 & \text{otherwise.}
\end{cases}
\ee
\end{lemma}
\begin{proof}
Considering the linear terms in~$\epsilon$ in Eq.~\eqref{eq:spectrum} (after setting $c=0$) yields
\be
\label{eq:temp}
\wh G_k^{[1]} s_\lambda\+\sum_{\nu\in\partitions}c(\lambda,\nu)\,\wh G_k^{[0]}s_\nu\=
24\,E^{[1]}_k(\lambda)\,s_\lambda\+Q_{k+2}(\lambda)\sum_{\nu\in\partitions}c(\lambda,\nu)\,s_\nu\,,
\ee
where we use Eqs.~\eqref{eq:E0}, \eqref{eq:compGhat} (with $c=0$) and~\eqref{eq:defa}.
Using Eq.~\eqref{eq:Dubrovin} we have $G_k^{[0]}s_\nu=Q_{k+2}(\nu)s_\nu$ and so taking the scalar product of Eq.~\eqref{eq:temp} with $s_\mu$ we obtain
\be
\label{eq:firstorderspectrum}
\langle s_\mu,\wh G_k^{[1]} s_\lambda\rangle\+Q_{k+2}(\mu)\,c(\lambda,\mu)\=Q_{k+2}(\lambda)\,c(\lambda,\mu).
\ee
Setting $k=1$, this shows the part of the statement relative to the case $Q_3(\lambda)\not=Q_3(\mu)$.

When $\lambda=\mu$ there is nothing to prove because $c(\lambda,\lambda)=0$ by Eq.~\eqref{eq:gauge}, as noted above.
Therefore, let us consider the only remaining case, which is $Q_3(\lambda)=Q_3(\mu)$ and $\lambda\neq\mu$.
In such case, Eq.~\eqref{eq:firstorderspectrum} with $k=1$ shows that $\langle s_\mu,\wh G_1^{[1]} s_\lambda\rangle=0$. Hence,  by~Lemma~\ref{lemma:scalarproduct}(\ref{it:scalarproduct-ii}) we find $\mu\not\in\mathcal U(\lambda)$. Therefore, by Lemma~\ref{lemma:scalarproduct}(\ref{it:scalarproduct-i}) we obtain $\langle s_\mu,\wh G_k^{[1]} s_\lambda\rangle=0$ for all $k\geq 0$.
Then, if there exists $k\geq 1$ such that $Q_k(\lambda)\not=Q_k(\mu)$ we can conclude by looking at Eq.~\eqref{eq:firstorderspectrum} that $c(\lambda,\mu)=0$, as claimed.
The existence of such $k\geq 1$ is obvious because we can map each partition~$\nu$ to the entire function in~$z$ given by
\begin{align}
\nonumber
\sum_{k\geq 1}\bigl(Q_k(\nu)-\beta_k\bigr)\,z^{k-1}&\:=\:\sum_{i\geq 1}\Bigl(\e^{z(\nu_i-i+\frac 12)}-\e^{z(-i+\frac 12)}\Bigr),
\end{align}
and this map is clearly injective.
\end{proof}

The previous lemmas imply that Theorem~\ref{thm:2} is equivalent to the following statement.

\begin{theorem}\label{thm:2b}
For all partitions~$\lambda$ 
\begin{align}
    \nonumber
r^{[1]}_\lambda\,= {}&\frac 1{12}
\sum_{\mu \in \mathcal{U}(\lambda)} (-1)^{a+a'+b+b'}\,\Bigl(\frac{\lambda_a-a-\mu_{a'}+{a'}}{\lambda_a-a-\mu_{b'}+b'}-\frac{\lambda_a-a-\mu_{b'}+b'}{\lambda_a-a-\mu_{a'}+a'}\Bigr) \,s_{\mu}\mspace{1mu}.
\end{align}
\end{theorem}
\begin{proof}
With the notation in Eq.~\eqref{eq:notationU}, we note that for $\mu\in\mathcal U(\lambda)$ we have
\begin{align}
\nonumber
Q_3(\lambda)-Q_3(\mu)\={}&
\frac 12\Bigl(\bigl(\lambda_a-a+\tfrac 12\bigr)^2+\bigl(\lambda_b-b+\tfrac 12\bigr)^2-\bigl(\mu_{a'}-a'+\tfrac 12\bigr)^2-\bigl(\mu_{b'}-b'+\tfrac 12\bigr)^2\Bigr)
\\
\={}&
\bigl(\lambda_a-a-\mu_{a'}+a'\bigr)\bigl(\lambda_a-a-\mu_{b'}+b'\bigr)\,,
\end{align}
as it follows from Eqs.~\eqref{eq:QkC} and~\eqref{eq:remarkforsimplification}.
(Note that this quantity also equals $w_{\lambda\mu}\,|\gamma_1|\,|\gamma_2|$, in the notation of Theorem~\ref{thm:2}.)
Then, from Lemma~\ref{lemma:scalarproduct} and Lemma~\ref{lemma:simple2} and elementary algebraic simplification using Eq.~\eqref{eq:remarkforsimplification}, we obtain the statement of the theorem. 
\end{proof}

\appendix
\renewcommand{\arraystretch}{1.5}

\section{Tables supporting Conjecture~\ref{conj}}\label{app:table}

Polynomials obtained by matching Conjecture~\ref{conj} with the eigenvalues computed in SageMath  for $k=0,\ldots, 10$. Note that we did not restrict the degree of $f_{D,\nu}$ to be~$2D$, but only verified this bound in each of the examples. 
\[\begin{array}{lll} 
  \nu   & D \hspace{5pt} & f_{D,\nu}(k) \hspace{323pt} \\\hline\hline
  \emptyset  & 0 & 1 \\\hline
  \emptyset  & 1 & \frac{1}{24} k (k+3) \\\hline
  \emptyset  & 2 & \frac{1}{3456} k (k+4) (3 k^2+16 k+17) \\\hline
  \emptyset  & 3 & \frac{1}{1244160} k (k+5) (15 k^4+210 k^3+896 k^2+1405 k+50)\\\hline
  \emptyset  & 4 & \frac{1}{597196800} k (k+6) (75 k^6+1950 k^5+17570 k^4+68042 k^3+89913 k^2 \\&&\hfill -100568 k-277382) \\\hline
  \emptyset & 5 & \frac{1}{300987187200} k (k+7) (315 k^8+13020 k^7+199920 k^6+1433012 k^5+4539665 k^4\\&& \hfill + 142128 k^3-39506516 k^2-99890840 k-59164224) \\ \hline
 \hline
 (1) & 0 & \frac{1}{12}\\\hline
 (1) & 1 & \frac{1}{288} \left(k^2+2 k-1\right)\\\hline
(1) & 2 & \frac{1}{207360} (k+1) \left(15 k^3+95 k^2-148 k-616\right)\\\hline
 (1) & 3 & \frac{1}{14929920} \left(15 k^6+240 k^5+467 k^4-5229 k^3-24476 k^2-33159 k-7482\right)\\\hline
(1) & 4 &  \frac{1}{{50164531200}}\bigl(525 k^8+14700 k^7+91910 k^6-448826 k^5-5919635 k^4\\&&\hfill-17900130 k^3-11934720 k^2+25552736 k+32402640\bigr)\\\hline
\hline
 (2) & 0 &\frac{1}{288}\\\hline
(2) & 1 &\frac{1}{103680}\left(15 k^2+323 k+278\right) \\\hline
(2) & 2 &\frac{1}{4976640} \left(15 k^4+696 k^3+3170 k^2+3631 k-30\right) \\\hline
(2) & 3  &\frac{1}{2508226560} \bigl(105 k^6+7833 k^5+85309 k^4+27485 k^3-1345564 k^2-3262016 k\\&&\hfill-1993680\bigr) \\\hline
\hline
\end{array}\]\[
\begin{array}{lll} 
 \nu   & D \hspace{5pt}& f_{D,\nu}(k) \hspace{318pt}  \\\hline\hline
(1^2) & 0 & \frac{1}{288} \\\hline
(1^2) & 1 & \frac{1}{6912}
k (k+1) \\\hline
(1^2) & 2 & \frac{1}{4976640}\left(15 k^4+80 k^3-421 k^2-486 k-60\right) \\\hline
(1^2) & 3 & \frac{1}{358318080}\left(15 k^6+195 k^5-832 k^4-9599 k^3-24029 k^2-22834 k-4056\right) \\\hline
\hline
(3) & 0 &-\frac{1}{6912} \\\hline
(3) & 1 & -\frac{1}{165888}k (k+1) \\\hline
(3) & 2 & \frac{1}{119439360}\left(-15 k^4-80 k^3+421 k^2+486 k+60\right) \\\hline
\hline
(2,1) & 0 &\frac{1}{3456}\\\hline
(2,1) & 1 &\frac{1}{1244160}\left(15 k^2+308 k+30\right)\\\hline
(2,1) & 2 &\frac{1}{59719680}\left(15 k^4+666 k^3+1690 k^2+2271 k+1052\right)\\\hline\hline
(1^3) & 0 &\frac{1}{10368}\\\hline
(1^3) & 1 & \frac{1}{248832}\left(k^2+3\right)\\\hline
(1^3) & 2 & \frac{1}{179159040} \left(15 k^4+50 k^3-699 k^2+994 k+60\right)\\\hline
\hline
(4) & 0 & -\frac{967}{829440}\\\hline
(4)  & 1 &\frac{1}{418037760}\left(-20307 k^2-290329 k-229408\right)\\\hline
\hline
(3,1) & 0 & -\frac{109}{622080} \\\hline
(3,1) & 1 & \frac{1}{29859840}\left(-218 k^2+1247 k-1488\right)\\\hline
\hline
(2^2)& 0 &\frac{1}{165888} \\\hline
(2^2)& 1 & \frac{1}{59719680}\left(15 k^2+601 k+1232\right)\\\hline
\hline
(2,1^2)& 0 & \frac{1}{82944}\\\hline
(2,1^2)& 1 &\frac{1}{29859840}\left(15 k^2+293 k-188\right) \\\hline
\hline
(1^4)& 0 & \frac{1}{497664}\\\hline
(1^4)& 1 & \frac{1}{11943936}\left(k^2-k+8\right)\\\hline
\hline
(5) & 0 &\frac{253}{2903040} \\\hline\hline
(4,1)  & 0 &\frac{967}{238878720} \\\hline\hline
(3,2)  & 0 & -\frac{109}{8599633920}\\\hline\hline
(3,1^2)  & 0 & \frac{109}{206391214080}\\\hline\hline
(2^2,1)  & 0  &\frac{1}{660451885056} \\\hline\hline
(2,1^3)  & 0 &-\frac{1}{23776267862016} \\\hline\hline
(1^5) & 0 &\frac{1}{5706304286883840} \\\hline\hline\\
\end{array}
\]

\end{document}